\documentclass[reqno]{amsart}%

\usepackage{amsmath, amsthm, amssymb}%
\usepackage{mathtools}%
\usepackage{enumitem}%
\usepackage{interval}%
\usepackage{booktabs}%
\usepackage{algorithm}%
\usepackage{algpseudocode}%
\usepackage{caption}%
\usepackage{subcaption}%
\usepackage[round]{natbib}%
\usepackage{doi}%
\usepackage{hyperref}%

\intervalconfig{soft open fences}%
\hypersetup{
  colorlinks   = true, %
  urlcolor     = blue, %
  linkcolor    = blue, %
  citecolor   = red %
}%
\mathtoolsset{showonlyrefs,showmanualtags}%
\DeclarePairedDelimiter{\floor}{\lfloor}{\rfloor}%
\DeclarePairedDelimiter{\ceil}{\lceil}{\rceil}%

\newtheorem{thm}{Theorem}[section]
\newtheorem{cor}[thm]{Corollary}
\newtheorem{lem}[thm]{Lemma}

\theoremstyle{remark}
\newtheorem{exmp}{Example}%

\theoremstyle{definition}

\newcommand{\Z}{\mathbb{Z}}%
\newcommand{\N}{\mathbb{N}}%
\newcommand{\Q}{\mathbb{Q}}%
\newcommand{\R}{\mathbb{R}}%

\DeclareMathOperator{\dbf}{dbf}%
\DeclareMathOperator{\rbf}{rbf}%
\DeclareMathOperator{\lcm}{lcm}%
\DeclareMathOperator{\sgn}{sgn}%
\newcommand{\bigO}{\mathcal{O}}%
\newcommand{\hyp}{\mathcal{T}}%
\newcommand{\poly}{\mathsf{poly}}%
\newcommand{\Thm}{T_{\mathrm{hm}}}%

\makeatletter
\newcommand{\algmargin}{\the\ALG@thistlm}
\makeatother

\newlength{\whilewidth}
\algdef{SE}[parWHILE]{parWhile}{EndparWhile}[1]
{\parbox[t]{\dimexpr\linewidth-\algmargin}{%
    \hangindent\whilewidth\strut\algorithmicwhile\ #1\
    \algorithmicdo\strut}}{\algorithmicend\ \algorithmicwhile}%

\newlength{\ifwidth}
    \algdef{SE}[parIF]{parIf}{EndparIf}[1]
    {\parbox[t]{\dimexpr\linewidth-\algmargin}{%
        \hangindent\ifwidth\strut\algorithmicif\ #1\
          \algorithmicthen\strut}}{\algorithmicend\ \algorithmicif}%

\algnewcommand{\parState}[1]{\State%
  \parbox[t]{\dimexpr\linewidth-\algmargin}{\strut#1\strut}}
\algnewcommand\Break{\textbf{break}}%

\renewcommand\u[1]{\underline{#1}}%

\title[Cutting-plane schedulability tests]{Cutting-plane algorithms for
  preemptive uniprocessor real-time scheduling problems}%

\author{Abhishek Singh}%
\date{}%
\address{Department of Computer Science \& Engineering\\
  Washington University in St.\ Louis\\
  MO 63130, U.S.A.}%
\email{abhishek.s@wustl.edu}%

\begin{document}

\maketitle%

\begin{abstract}
  Fixed-point iteration algorithms like RTA (response time analysis) and QPA
  (quick processor-demand analysis) are arguably the most popular ways of
  solving schedulability problems for preemptive uniprocessor FP
  (fixed-priority) and EDF (earliest-deadline-first) systems. Several IP
  (integer program) formulations have also been proposed for these problems, but
  it is unclear whether the algorithms for solving these formulations are
  related to RTA and QPA\@. By discovering connections between the problems and
  the algorithms, we show that RTA and QPA are, in fact, suboptimal
  cutting-plane algorithms for specific IP formulations of FP and EDF
  schedulability, where optimality is defined with respect to convergence rate.
  We propose optimal cutting-plane algorithms for these IP formulations. We
  compare the new algorithms with RTA and QPA on large collections of synthetic
  systems to gauge the improvement in convergence rates and running times.\\
  \smallskip%

  \noindent{}\keywordsname: hard real-time scheduling, fixed priority, earliest deadline
  first, cutting planes, linear programming duality, fixed-point iteration
\end{abstract}

\section{Introduction}%
\label{sec:intro}%

FP (fixed-priority) and EDF (earliest-deadline-first) are two popular methods of
assigning priorities to preemptible hard real-time tasks in a uniprocessor
priority-driven system. In FP systems, each task is assigned a priority that
remains constant during execution; in contrast, priorities of tasks in EDF
systems are variable during execution, and at any instant a task with the
earliest deadline has the highest priority. From the early work of
\citet{liuSchedulingAlgorithmsMultiprogramming1973}, it has been known that the
two systems are quite different: for instance, an implicit-deadline FP system
can violate its timing requirements even when processor utilization is as low as
70\%, while a comparable EDF system is safe for all processor utilizations up to
100\%.\footnote{Implicit deadlines are defined in Section~\ref{sec:models}.}
Preemptive uniprocessor systems with FP and EDF priority assignments are some of
the most well-studied systems in hard real-time scheduling theory: historical
perspectives on these systems and other supporting literature may be found in
surveys, handbooks, and textbooks~\citep{audsleyFixedPriorityPreemptive1995,
  liuRealTimeSystems2000, shaRealTimeScheduling2004,
  buttazzoHardRealTimeComputing2011, levyHandbookRealTimeComputing2020}; works
that study the differences between FP and EDF systems are also
available~\citep{buttazzoRateMonotonicVs2005,
  rivasSchedulabilityAnalysisOptimization2011, peraleRemovingBiasJudgment2021}.

Systems that do not violate their timing requirements are called safe or
\emph{schedulable}; other systems are said to be \emph{unschedulable}. A
schedulable system is often characterized by a \emph{schedulability condition}:
for instance, a schedulable constrained-deadline preemptive uniprocessor FP
system with $n$ tasks listed in nonincreasing order of priority is characterized
by the existence of a $t \in \linterval{0}{D_i}$ such that $\rbf_i(t) \le t$ for
all $i \in [n]$, where $D_i$ is the relative deadline of the task $i$ and
$\rbf_i$ maps any $t$ to the maximum amount of work generated by the subsystem
$[i]$ in any interval of length $t$~\citep{josephFindingResponseTimes1986,
  lehoczkyRateMonotonicScheduling1989}.\footnote{$[n]$ is shorthand for
  $\{1,2,\ldots,n\}$; we assume that $[0] = \emptyset$.} An unschedulable
arbitrary-deadline preemptive uniprocessor EDF system with $n$ tasks can also be
characterized by the existence of a $t \in \linterval{0}{\lcm_{j \in [n]} T_j}$
such that $\dbf(t) > t$ where $T_j$ is the period of task $j$ and $\dbf$ maps
any $t$ to the amount of work generated by the system in the interval that must
be completed in the
interval~\citep{baruahPreemptivelySchedulingHardrealtime1990,
  ripollImprovementFeasibilityTesting1996}. We will discuss schedulability
conditions, relative deadlines, constrained deadlines, arbitrary deadlines,
$\rbf$, periods and $\dbf$ in more detail in a later section. For now, it
suffices to know that the schedulability conditions for both FP and EDF systems
are about the existence of a $t$ in a bounded interval where some function of
$t$ is nonnegative; for FP (resp., EDF) systems, the function is $t \mapsto t -
\rbf_i(t)$ (resp., $t \mapsto \dbf(t) - t - 1$).

Given a description of a hard real-time system, the problem of deciding whether
the given system satisfies its timing requirements is called a
\emph{schedulability problem}. An algorithm that solves a schedulability problem
is called a \emph{schedulability test}. Given a system, a schedulability test
checks whether the appropriate schedulability condition holds for the system.

An FP schedulability test attempts to find a $t \in \linterval{0}{D_i}$ where
$\rbf_i(t) \le t$. There are many algorithmic techniques that can be used to
achieve this goal: for example, we can use fixed-point
iteration~\citep{josephFindingResponseTimes1986,
  audsleyApplyingNewScheduling1993}, depth-bounded search
trees~\citep{manabeFeasibilityDecisionAlgorithm1998,
  biniSchedulabilityAnalysisPeriodic2004}, and continued
fractions~\citep{parkDeterminingRateMonotonic2023}. The fixed-point iteration
test is called RTA (response time analysis), and the depth-bounded search tree
test is called HET (hyperplanes exact test). The running times of these tests
are, in general, incomparable: the most significant factor in the worst-case
running time of RTA is $D_{\max} / T_{\min}$ where $D_{\max}$ (resp.,
$T_{\min}$) is the largest deadline (resp., the smallest period) in the system;
the most significant factor in the worst-case running time of HET is
$2^{\text{\#periods}}$. Thus, for hard problem instances where the number of
periods is small but $D_{\max} / T_{\min}$ is large, HET will likely outperform
RTA\@; on the other hand, for hard problem instances where the number of periods
is large and $D_{\max} / T_{\min}$ is small, RTA will likely outperform HET\@.

Similarly, given a system that is not trivially unschedulable an EDF
schedulability test attempts to find a $t \in \ointerval{0}{\lcm_{j \in [n]}
  T_j}$ such that $\dbf(t) > t$. Many algorithmic approaches can be used to
search for $t$ including fixed-point
iteration~\citep{zhangSchedulabilityAnalysisRealTime2009}, integer programming
in fixed dimension~\citep{baruahAlgorithmsComplexityConcerning1990} (utilizing
Lenstra's algorithm~\citep{lenstraIntegerProgrammingFixed1983}), and convex-hull
computation~\citep{biniCuttingUnnecessaryDeadlines2019}. The fixed-point
iteration algorithm is called QPA (quick processor-demand analysis). As was the
case with FP schedulability tests, the running times of EDF schedulability tests
are not comparable, in general: the most significant factor in the worst-case
running time of QPA is $\lcm_{j \in [n]} T_j / T_{\min}$; the most significant
factor in the worst-case running time of Lenstra's algorithm is $p^{\bigO(p)}$
where $p$ is the variety of the system.\footnote{The variety of a synchronous
  system is the number of distinct deadline-period pairs in the
  system~\citep{baruahFixedParameterAnalysisPreemptive2022}. Synchronous and
  asynchronous systems are defined in Section~\ref{sec:models}.} Thus, for hard
problem instances where the variety is very small and $\lcm_{j \in [n]} T_j /
T_{\min}$ is large, Lenstra's algorithm will likely outperform QPA\@; in
contrast, for hard problem instances where the ratio is small and the variety is
large, QPA will likely outperform Lenstra's algorithm. Differences in worst-case
running times for various algorithmic approaches for FP and EDF schedulability
tests with respect to the sizes of different problem parameters such as the
number of periods in the system, the variety of the system, $D_{\max} /
T_{\min}$, and $\lcm_{j \in [n]} T_j / T_{\min}$ have been studied recently
using the framework of parameterized algorithms by
\citet{baruahFixedParameterAnalysisPreemptive2022}.

Although applying new algorithmic techniques to create new schedulability tests
with better properties is very important, it is equally important to attempt to
understand the relationships between the schedulability problems and between the
algorithms for the problems. Efforts in the latter direction often yield new
structural and algorithmic insights and result in a coherent unified theory. We
are interested in studying the connections between the FP schedulability problem
and the EDF schedulability problem in the context of four algorithmic
approaches:
\begin{description}
\item[RTA] the fixed-point iteration algorithm for FP schedulability (see
  Section~\ref{sec:rta}).
\item[IP-FP] the IP (integer programming) formulation for FP schedulability
  where the variables correspond to the integral quantities in $\rbf_i$ (see
  Section~\ref{sec:ip_fp}).
\item[QPA] the fixed-point iteration algorithm for EDF schedulability (see
  Section~\ref{sec:qpa}).
\item[IP-EDF] the IP formulation for EDF schedulability where the variables
  correspond to the integral quantities in $\dbf$ (see
  Section~\ref{sec:ip_edf}).
\end{description}
While RTA and QPA are well-known, IP-FP and IP-EDF are nonstandard names that we
are using to refer to specific IP formulations of the schedulability problems
described later in the document; moreover, IP-FP and IP-EDF qualify only vaguely
as algorithmic approaches because we have not specified an algorithm for solving
the IP formulations yet (we will design a new cutting-plane algorithm to solve
the IPs in a unified manner).

It is natural to try to understand the relationships within the pairs (RTA,
IP-FP) and (QPA, IP-EDF) because they solve the FP schedulability problem and
the EDF schedulability problem respectively. However, we found that studying the
problems separately is a little inefficient because both problems can be reduced
in polynomial time to a common problem called \emph{the kernel} (see
Section~\ref{sec:kern}):

\begin{figure}[h]
  \centering%
  \includegraphics[width=0.65\linewidth]{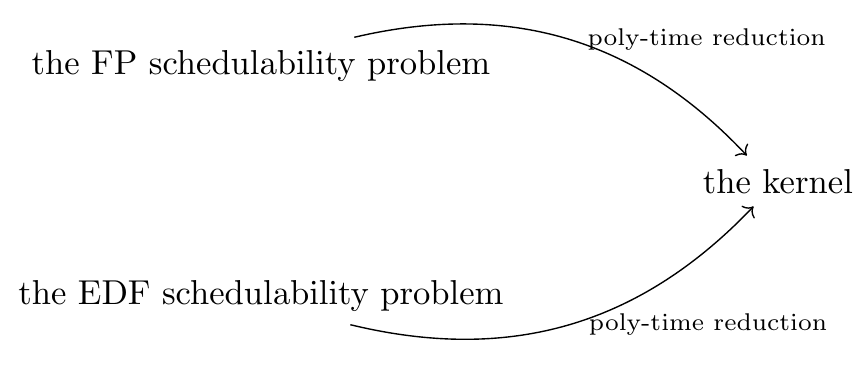}%
  \label{fig:rel}%
\end{figure}

If the kernel can be solved efficiently by some algorithm, then both FP
schedulability and EDF schedulability can be solved efficiently; thus, the
kernel captures the essence of the hardness of the two
problems.\footnote{Reductions from EDF schedulability to FP schedulability have
  been used by \citet{ekbergFixedPrioritySchedulabilitySporadic2017} to prove
  the hardness of FP schedulability using the hardness of EDF schedulability.
  Thus, the idea that the two problems are closely related is not new. However,
  the kernel and the reductions to it have not been described before, to the
  best of our knowledge.} We prove the following facts about the kernel:
\begin{itemize}
\item The kernel can be solved by a fixed-point iteration algorithm called
  FP-KERN (see Section~\ref{sec:fp_kern}). RTA and QPA can be recovered from
  FP-KERN by composing it with the above reductions.
\item The kernel has an IP formulation called IP-KERN (see
  Section~\ref{sec:ip_kern}). IP-FP and IP-EDF can be recovered from IP-KERN by
  composing it with the above reductions.
\end{itemize}
Since any relationship between FP-KERN and IP-KERN must also exist for RTA
(resp., QPA) and IP-FP (resp., IP-EDF), we can focus our attention on FP-KERN
and IP-KERN\@.

Our main result is that IP-KERN can be solved by a family of cutting-plane
algorithms; FP-KERN is a member of the family but it is \emph{not} the optimal
algorithm in this family with respect to \emph{convergence rate}, which is
defined as the inverse of the number of iterations required for convergence; and
the optimal algorithm in the family, called CP-KERN, has a better convergence
rate than FP-KERN\@. Viewed from the FP (resp., EDF) perspective, RTA (resp.,
QPA) is a suboptimal cutting-plane algorithm and CP-KERN converges to the
solution in fewer iterations. In our empirical evaluation, we compare the number
of iterations required by the fixed-point iteration algorithms (RTA and QPA) and
CP-KERN on synthetic systems; the results confirm that CP-KERN has a better
convergence rate than RTA and QPA\@.

\subsection{Practical concerns: computational complexity}

Faster schedulability algorithms are needed in many applications. In holistic
analyses of distributed real-time systems, schedulability tests are called a
great many times till the values for all parameters in the system stabilize or
an unschedulable subsystem is
discovered~\citep{tindellHolisticSchedulabilityAnalysis1994,
  spuriHolisticAnalysisDeadline1996}. In partitioned approaches to
multiprocessor scheduling, uniprocessor schedulability tests are used to
determine the feasibility of (usually numerous) partitions. In automatic or
interactive design space explorations to optimize objectives such as energy
consumption, schedulability tests are used to determine the feasibility of
numerous configurations. The speed of a schedulability test is also crucial when
it is used as an online test in a dynamic embedded system.

While the convergence rate of CP-KERN is better than FP-KERN (think, RTA and
QPA), this does not imply that CP-KERN is faster than FP-KERN\@. Consider, for
instance, the center of gravity method for convex programs~\citep[see, for
instance, ][Sec.\ 2.1]{bubeckConvexOptimizationAlgorithms2015}. The center of
gravity method has a good convergence rate, but it requires the center of
gravity of a convex body to be computed in each iteration for which no efficient
procedures are known. CP-KERN, unlike the center of gravity method, is not
purely theoretical. In each iteration of CP-KERN we must solve a linear
relaxation of IP-KERN\@. Since linear programs can be solved in polynomial time
by the ellipsoid method~\citep{khachiyanPolynomialAlgorithmsLinear1980} and
interior point methods~\citep{karmarkarNewPolynomialtimeAlgorithm1984}, each
iteration of CP-KERN has polynomial running time.\footnote{The ``running time''
  of an algorithm in a theoretical context refers to the number of elementary
  arithmetic operations used by the algorithm.}

To solve linear relaxations of IP-KERN in each iteration of CP-KERN even more
efficiently, we propose a specialized method that runs in strongly polynomial
time.\footnote{An algorithm runs in strongly polynomial time if the algorithm is
  a polynomial space algorithm and uses a number of arithmetic operations which
  is bounded by a polynomial in the number of input numbers~\citep[see][pg.
  32]{grotschelGeometricAlgorithmsCombinatorial1993}.} Since no strongly
polynomial-time algorithms are known for linear programming, our specialized
method is an improvement over using a general algorithm for solving linear
programs. Moreover, we show that if CP-KERN uses this specialized method then it
has the same worst-case running time as FP-KERN\@. This result, combined with
the optimality with respect to convergence rate, establishes CP-KERN as a better
algorithm than FP-KERN in theory. In practice, FP-KERN may be faster than
CP-KERN for instances where the difference in convergence rates of CP-KERN and
FP-KERN is negligible because FP-KERN does less work in each iteration than
CP-KERN when we do not ignore constant factors. We compare the running times on
synthetic systems in our empirical evaluation.

\subsection{Practical concerns: implementation complexity}

Since FP-KERN proceeds by fixed-point iteration, it is quite easy to implement
without depending on external libraries; thus, it is well-suited to serve as an
online schedulability test in a dynamic embedded system where task and resource
constraints are subject to change over time. Since a cutting-plane algorithm for
an IP solves a linear program in each iteration, an implementation of a
cutting-plane algorithm for solving general IPs usually depends on linear
programming libraries, commercial or otherwise, making it harder to deploy on
dynamic embedded systems.\footnote{Projects like CVXGEN attempt to address the
  problem of deploying convex programming solvers on embedded systems by
  generating custom C code for some families of convex
  programs~\citep{mattingleyCVXGENCodeGenerator2012}.} When CP-KERN utilizes our
specialized algorithm to solve relaxations of IP-KERN\@, it is about 50 lines in
pseudocode, it uses elementary arithmetic operations, and does not depend on
mathematical programming or algebra libraries in function calls. Therefore,
CP-KERN has a small footprint and is well-suited for deployment on dynamic
embedded systems. For software developers writing a schedulability library, our
structured approach promotes code reuse and reduces development effort.
Moreover, the smaller codebase inspires more trust in stakeholders.

\section{System Models}%
\label{sec:models}%

In a hard real-time system, a task receives an infinite stream of requests and
it must respond to each request within a fixed amount of time. We consider a
system comprising $n$ independent preemptible sporadic tasks, labeled
$1,2,\dotsc,n$. We refer to the system simply as $[n]$; given a system $[n]$, $i
\in [n]$ is the $i$-th task, and $[i] \subseteq [n]$ is a subsystem of $[n]$
that contains tasks $\{1,2,\ldots,i\}$. Each sporadic task $i \in [n]$ has the
following (integral) characteristics:
\begin{table}[H]
  \setlength{\tabcolsep}{0.6\tabcolsep}%
  \centering%
  \begin{tabular}{c p{0.8\linewidth}}\toprule
    Symbol &Task Characteristic\\\midrule
    $C_i$ & worst-case execution time (\emph{wcet})\\
    $T_i$ & minimum duration between successive request arrivals (\emph{period})\\
    $D_i$ & maximum duration between request arrival and response (\emph{relative deadline})\\
    $J_i$ & maximum duration between request arrival and the task becoming eligible for execution (\emph{release jitter})\\\bottomrule
  \end{tabular}
\end{table}

\noindent{}If the $k$-request for task $i$ arrives at time $t$, then the
following statements must be true:
\begin{itemize}
\item The request must be completed by time $t + D_i$ in a feasible schedule;
  $t+D_i$ is the \emph{absolute deadline} for the completion of the response to
  the request.
\item The $(k+1)$-th request for task $i$ cannot arrive before $t+T_i$.
\end{itemize}
Sometimes a request may arrive at time $t$ and the task may not become eligible
for execution until time $t+\delta$ where $\delta > 0$; $\delta$ is the release
jitter experienced by task $i$ at time $t$, and $J_i$ is the maximum release
jitter that can be experienced by task $i$.

If $D_i = T_i$, then the deadline $D_i$ is said to be implicit; if $D_i \le
T_i$, then the deadline $D_i$ is said to be constrained. If all tasks in a
system have implicit deadlines, then the system is an \emph{implicit-deadline
  system}. If all tasks in a system have constrained deadlines, then the system
is a \emph{constrained-deadline system}; otherwise, it is an
\emph{arbitrary-deadline system}.

We refer to the largest (resp., smallest) period in the system by $T_{\max}$
(resp., $T_{\min}$); similar symbols are used for other task parameters as well.
We refer to the harmonic mean of the periods by $\Thm$.

As mentioned in Section~\ref{sec:intro}, we assume that the tasks are
preemptible, run on a single processor, and are scheduled by an FP or EDF
scheduler.

\subsection{FP schedulability and RTA}%
\label{sec:rta}%

We restrict our attention to synchronous constrained-deadline preemptive
uniprocessor FP systems, and we simply call them FP systems. We assume that
tasks are listed in decreasing order of priority in FP systems.

An FP system $[n]$ is schedulable if and only if the subsystem $[n-1]$ is
schedulable and the condition
\begin{equation}%
  \label{cond:fp}%
  \exists t \in \linterval{0}{D_n-J_n}: \rbf_n(t) \le t
\end{equation}
holds~\citep{josephFindingResponseTimes1986,
  lehoczkyRateMonotonicScheduling1989, audsleyApplyingNewScheduling1993}. Here,
$\rbf_i$, the request bound function of subsystem $[i] \subseteq [n]$, is given
by
\begin{equation}%
  \label{def:rbf}%
  t \mapsto \sum_{j \in [i]} \left\lceil \frac{t + J_j}{T_j} \right\rceil C_j.
\end{equation}
From now on, we refer to $\rbf_n$ simply as $\rbf$.

Condition~\eqref{cond:fp} is satisfied if and only if the problem
\begin{equation}%
  \label{opt:rta}%
  \min \{t \in \linterval{0}{D_n-J_n} \cap \Z \mid\, \rbf(t) \le t\}
\end{equation}
has an optimal solution.\footnote{Since there is only one variable $t$ in the
  problem and the minimization objective is also $t$, if an optimal solution
  exists then it is also unique. Therefore, we can say ``the optimal solution''
  instead. However, since we will look at IP formulations of the problem later
  which will allow more than one optimal solution, we stick to the phrase ``an
  optimal solution''.} It can be shown that the optimal solution $t^*$, if it
exists, is the smallest $t \in \linterval{0}{D_n-J_n}$ that satisfies
\[
  \rbf(t) = t,
\]
i.e., $t^*$ is the smallest fixed point of $\rbf$. $t^*$ can be found by using
fixed-point iteration, i.e., by starting with a safe lower bound $\u{t}$ for the
fixed point and iteratively updating $\u{t}$ to $\rbf(\u{t})$. This algorithm is
called RTA (response time analysis) because $t^*+J_n$ is the worst-case response
time of task $n$. Theorem~\ref{thm:fixed} can be applied to derive the following
theorem; see also the discussion preceding Theorem~\ref{thm:fp-kern}.

\begin{thm}\label{thm:rta}
  Problem~\eqref{opt:rta} can be solved by RTA in
  \[
    \bigO\left(\frac{(D_n-J_n)n^2}{\Thm}\right)
  \]
  running time.
\end{thm}

\begin{exmp}
  Consider the FP task system with implicit deadlines and zero release jitter
  shown in Table~\ref{tab:ex1}. The execution of RTA for this system is depicted
  in Figure~\ref{fig:rta} from the initial value $\u{t} = 20$. $\rbf$ for the
  system is shown as the dashed blue step function, the computation of
  $\rbf(\u{t})$ is shown as upward red arrows, and the update $\u{t}
  \gets \rbf(\u{t})$ is shown as rightward red arrows.
  \begin{table}
    \centering%
    \begin{tabular}{c c c}\toprule
      $i$ & $C_i$ & $T_i$\\\midrule
      $1$ & $20$ & $40$\\
      $2$ & $10$ & $50$\\
      $3$ & $33$ & $150$\\\bottomrule
    \end{tabular}%
    \caption{An FP task system with implicit deadlines and zero release
      jitter.}%
    \label{tab:ex1}%
  \end{table}
\end{exmp}

\begin{figure}
  \centering%
  \includegraphics[width=0.85\linewidth]{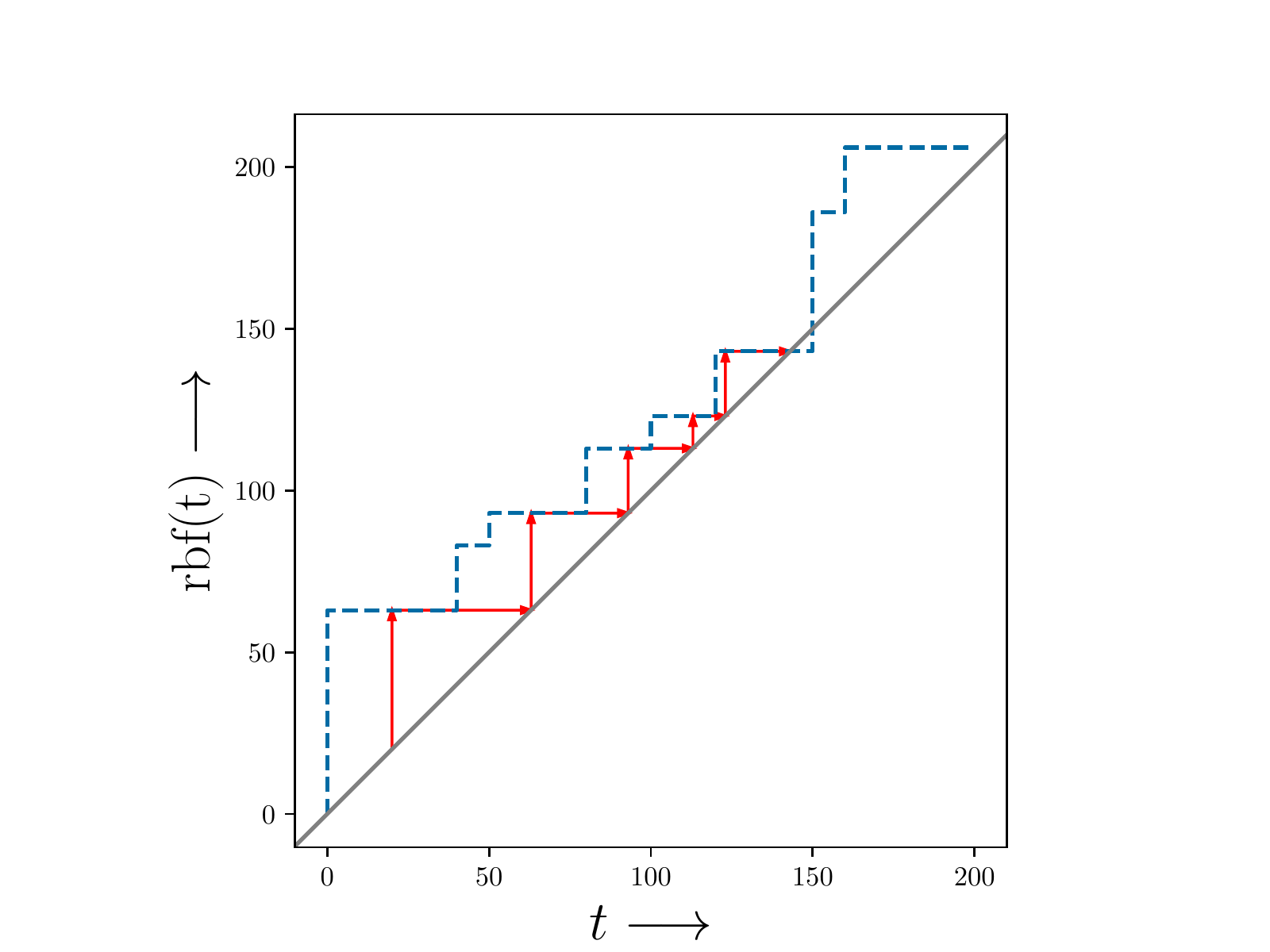}%
  \caption{$\rbf(t)$ (blue, dashed) for the system in Table~\ref{tab:ex1} and
    progress of RTA (red, solid arrows) from initial value $20$ (upward arrows
    denote the computation of $\rbf(\u{t})$ and the rightward arrows
    denote the update $\u{t} \gets \rbf(\u{t})$.).}%
  \label{fig:rta}%
\end{figure}

\subsection{EDF schedulability and QPA}%
\label{sec:qpa}%

We limit ourselves to arbitrary-deadline preemptive uniprocessor EDF systems,
and simply call them EDF systems. We assume that tasks are listed in
nondecreasing order of $\hat{D}_i - T_i$ in EDF systems, where
\[
  \hat{D}_i = D_i - J_i
\]

An EDF system is unschedulable if and only if $\sum_{j \in [n]} U_j > 1$ or
\begin{equation}%
  \label{cond:edf}%
  \exists t \in \rinterval{\hat{D}_{\min}}{L}: t < \dbf(t),
\end{equation}
holds~\citep{baruahPreemptivelySchedulingHardrealtime1990,
  ripollImprovementFeasibilityTesting1996, spuriAnalysisDeadlineScheduled1996,
  zhangSchedulabilityAnalysisEDFscheduled2013}. Here, $\dbf$, the demand bound
function of the system $[n]$, is given by
\begin{equation}%
  \label{def:dbf}%
  t \mapsto \sum_{j \in [n], t \ge \hat{D}_j - T_j} \left\lfloor\frac{t + T_j - \hat{D}_j}{T_j}\right\rfloor C_j,
\end{equation}
and $L$ is a large constant like $\lcm_j T_j$ (more details on $L$ will be
provided shortly).

The above condition is satisfied if and only if the problem
\begin{equation}%
  \label{opt:qpa}%
  \max \{t \in \rinterval{\hat{D}_{\min}}{L} \cap \Z: t < \dbf(t)\}
\end{equation}
has an optimal solution. It can be shown that the optimal solution $t^*$, if it
exists and is not trivially equal to $L-1$, is the largest $t \in
\rinterval{\hat{D}_{\min}}{L}$ that satisfies
\[
  \dbf(t) - 1 = t,
\]
i.e., $t^*$ is the largest fixed point of $\dbf(\cdot) - 1$. $t^*$ can be found
by starting a safe upper bound $\overline{t}$ for the fixed point and
iteratively updating $\overline{t}$ to $\dbf(\overline{t}) - 1$. This algorithm
is called QPA\@. The inventors of QPA use a slightly more elaborate iterative
step~\citep[Alg. 2]{zhangSchedulabilityAnalysisEDFscheduled2013}, but the above
update works assuming that all data are integral. Theorem~\ref{thm:fixed} can be
applied to derive the following theorem; see the discussion preceding
Theorem~\ref{thm:fp-kern}.

\begin{thm}\label{thm:qpa}
  Problem~\eqref{opt:qpa} can be solved by QPA in
  \[
    \bigO\left(\frac{(L - \hat{D}_{\min})n^2}{\Thm}\right)
  \]
  running time.
\end{thm}

Some similarities and differences between RTA and QPA are summarized in
Table~\ref{tab:comp}.

\begin{table}[t]
  \centering%
  \setlength{\tabcolsep}{3pt}%
  \begin{tabular}{l p{0.42\linewidth} p{0.40\linewidth}}\toprule
    &RTA &QPA \\\midrule
    system &synchronous, preemptive, uniprocessor, constrained-deadline, FP &synchronous, preemptive, uniprocessor, arbitrary-deadline, EDF\\
    iteration &update $t$ to $\rbf(t)$ &update $t$ to $\dbf(t) - 1$\\
    final value &least fixed point of $\rbf(\cdot)$, if system is schedulable &greatest fixed point of $\dbf(\cdot) - 1$, if system is unschedulable\\\bottomrule
  \end{tabular}
  \caption{Comparing RTA and QPA.}%
  \label{tab:comp}%
\end{table}

\subsection{The upper bound $L$}%
\label{sec:L}%

The interval of interest in Problem~\eqref{opt:qpa} is
$\rinterval{\hat{D}_{\min}}{L}$. The hyper-period of the system, denoted $\hyp$,
equals $\lcm \{ T_i \mid i \in [n] \}$ and can be used as $L$. A stronger bound
\begin{equation}%
  \label{def:La}%
  L_a = \min \{t \in \linterval{0}{\hyp} \mid\, \rbf(t) \le t\}.
\end{equation}
can be used as $L$ but it is harder to compute than the hyper-period. Finally,
if $\sum_{j \in [n]} U_j < 1$, then
\begin{equation}%
  \label{def:Lb}%
  L_b = \max\left\{\max_{j \in [n]} \{\hat{D}_j - T_j\}, \frac{\sum_{j \in [n]} (T_j - \hat{D}_j)U_j}{1 - \sum_{j \in [n]} U_j}\right\}
\end{equation}
can also be used as $L$. Many researchers have contributed to these bounds; more
details are provided, for instance,
by~\citet{georgePreemptiveNonPreemptiveRealTime1996} and
\citet{zhangSchedulabilityAnalysisEDFscheduled2013}.

\section{Background}%
\label{sec:back}%

\subsection{Fixed-point iteration}%
\label{sec:fpi}%

A fixed point of a function $f$ is a point $x$ such that
\[
  f(x) = x
\]
Given an initial approximation $x_0$ of a fixed point of $f$, \emph{fixed-point
  iteration} involves repeatedly applying $f$ to generate the sequence:
\[
  f(x_0), f(f(x_0)), f(f(f(x_0))), \dotsc
\]
In numerical analysis, the function is often a continuous real-valued function,
and the iteration is terminated when the last two generated values are within
some tolerance~\citep[see, for instance,][pg. 60]{burdenNumericalAnalysis2011}.
We will use fixed-point iteration only for monotonic step functions with a
finite number of steps in any bounded interval. For such functions, fixed-point
iteration can be shown to have the following property.

\begin{thm}\label{thm:fixed}\citep[][Thm. 1]{sjodinImprovedResponsetimeAnalysis1998}
  Let $\interval{a}{b}$ be an interval and let $f$ be a monotonically
  nondecreasing step function on $\interval{a}{b}$ with at most $n$ steps. One
  of the following statements must be true:
  \begin{enumerate}
  \item $f(a) \le a$.
  \item Fixed-point iteration with initial value $a$ terminates in at most $n$
    steps; if fixed points of $f$ lie in $\interval{a}{b}$, then the iteration
    converges to the smallest fixed point in $\interval{a}{b}$, otherwise a
    value larger than $b$ is generated.
  \end{enumerate}
\end{thm}

Although the above theorem was originally proved in the context of RTA (see
Section~\ref{sec:rta}), it also works more generally.

\subsection{Integer programs}%
\label{sec:ip}%

An optimization problem on the variable $x$ that can be expressed as
\begin{equation}%
  \label{lp:gen}%
  \min \{c\cdot x \mid x \in \R^n_{\ge 0}, Ax \ge b\},
\end{equation}
for some $m \times n$ matrix $A \in \R^{m \times n}$, $b \in \R^{m}$, $c \in
\R^n$ is a linear program. Here $c\cdot x$ is the dot product of the vectors.
LPs can be solved using methods such as the Fourier-Motzkin elimination method,
the simplex method~\citep{dantzigOriginsSimplexMethod1987}, the ellipsoid
method~\citep{khachiyanPolynomialAlgorithmsLinear1980}, and interior point
methods~\citep{karmarkarNewPolynomialtimeAlgorithm1984}. The last two methods
are known to run in polynomial time, and thus have theoretical significance.

If $x$ is restricted to be integral, then the problem is an \emph{integer
  program} (IP):
\begin{equation}%
  \label{ip:gen}%
  \min \{c\cdot x \mid x \in \Z^n_{\ge 0}, Ax \ge b\}.
\end{equation}
A \emph{relaxation} of an optimization problem with a maximization (resp.,
minimization) objective is a simpler optimization problem with optimal value at
least as large (resp., small) as the optimal value of the original problem.
Problems can be relaxed in many ways, but we will primarily be concerned with
relaxations of integer programs where the set of feasible solutions is expanded
by
\begin{itemize}
\item allowing variables to be continuous instead of integral; and/or
\item dropping constraints.
\end{itemize}
If the relaxation only involves making the variables continuous, then it is
called a linear relaxation. Thus, program~\eqref{lp:gen} is a linear relaxation
of program~\eqref{ip:gen}.

\subsection{Cutting-plane methods}%
\label{sec:cp}%

IPs can be solved using the cutting-plane methodology, which consists of
repeating the following steps:
\begin{enumerate}
\item Solve a relaxation of the IP\@; usually, the linear relaxation of the IP
  is chosen.
\item If the relaxation is found to be infeasible, then the IP is infeasible and
  we return \textrm{``infeasible''}.
\item If the optimal solution for the relaxation is integral, then it is an
  optimal solution for the IP, and we return the solution.
\item If the optimal solution for the relaxation is not integral, then find one
  or more \emph{cuts} or \emph{valid inequalities} that hold for all points in
  the convex hull of the set of feasible solutions for the IP but do not hold
  for the current solution. Since the cut separates the current solution from
  the feasible solutions, the problem of finding cuts is called the
  \emph{separation problem}.
\item Add the cuts to the IP and go to the first step.
\end{enumerate}
After each iteration, we have a better approximation for the convex hull of the
set of feasible solutions for the IP, but the problem description, in general,
is longer. Since the problem is more restricted in each iteration, the optimal
values for the relaxations generated in successive iterations of the
cutting-plane algorithm, denoted $\u{Z}_1, \u{Z}_2, \dotsc$, must satisfy
\[
  \u{Z}_1 \le \u{Z}_2 \le \dotsb
\]
assuming that the direction of optimization is minimization. These values are
sometimes called \emph{dual bounds} for the optimal value for the IP\@. More
details about integer programs can be found in standard
textbooks~\citep{alma991020577069705251}.

\subsection{Linear programming duality}%
\label{sec:duality}%

The linear program
\begin{equation}%
  \label{lp:gen_dual}%
  \max \{b\cdot y \mid y \in \R^m_{\ge 0}, A^Ty \le c \}
\end{equation}
is called the \emph{dual} of LP~\eqref{lp:gen}; LP~\eqref{lp:gen} is called the
primal program in such a context. Linear programming duality is the idea that
for such a pair of programs exactly one of the following statements is true:
\begin{itemize}
\item Both programs are infeasible.
\item One program is unbounded, and the other program is infeasible.
\item Both programs have optimal solutions with the same objective value.
\end{itemize}
The two programs have optimal solutions $x^*$ and $y^*$ with the same objective
value if and only if they satisfy \emph{complementary slackness} conditions:
\begin{itemize}
\item For all $i \in [m]$, $x^*$ satisfies the $i$-th constraint of
  LP~\eqref{lp:gen} with equality or $y^*_i = 0$.
\item Similarly, for all $i \in [n]$, $y^*$ satisfies the $i$-th constraint of
  LP~\eqref{lp:gen_dual} with equality or $x^*_i = 0$.
\end{itemize}

More details about linear programming duality can be found in standard
textbooks~\citep{matousekUnderstandingUsingLinear2007}.

\section{Integer programs for FP and EDF schedulability}

Here, we consider IP formulations for Problem~\eqref{opt:rta} and
Problem~\eqref{opt:qpa}; recall from Section~\ref{sec:models} that these
problems are essentially equivalent to the FP and EDF schedulability problems.

\subsection{IP-FP}%
\label{sec:ip_fp}%

Consider the following IP\@:
\begin{equation*}%
  \begin{array}{rrclcl}
    \min & \multicolumn{3}{l}{t} \\
    \textrm{s.t.}&t &\ge &1\\
         &t &\le &D_n-J_n\\
         &t - \sum_{j \in [n]} C_jx_j&\ge &0\\
         &T_ix_i - t&\ge &J_i,&& i \in [n]\\
         &t &\in &\Z\\
         &x &\in &\Z^n
  \end{array}
\end{equation*}
We show that this IP is equivalent to Problem~\eqref{opt:rta} in the next
theorem, and we call it IP-FP\@.

\begin{thm}%
  \label{thm:ip_fp}%
  Problem~\eqref{opt:rta} has a solution $t^*$ if and only if IP-FP has a
  solution $(t^*,x^*)$ for some $x^* \in \Z^n$.
\end{thm}
\begin{proof}[Proof Sketch]
  Any optimal solution $(t^*,x^*)$ to the IP satisfies the following identities:
  \begin{align*}
    t^* &= \max\left\{\sum_{j \in [n]} C_jx^*_j, 1\right\}\\
    x^*_i &= \ceil[\Bigg]{\frac{t^*+J_i}{T_i}},\quad i \in [n]
  \end{align*}
  Since $t^* \ge 1$ and $J_i \ge 0$, $x^*_i \ge 1$. Then, using the first
  identity and the fact that $C_i > 0$ for all $i$, we have $t^* \ge \sum_{j \in
    [n]} C_j > 1$. Then, the first identity can be simplified to
  \[
    t^* = \sum_{j \in [n]} C_jx^*_j = \sum_{j \in [n]}
    \ceil[\Bigg]{\frac{t^*+J_j}{T_j}} C_j = \rbf(t^*)
  \]
  The third equality uses the definition of $\rbf$ (Definition~\eqref{def:rbf}).
  Since the objective is to minimize $t$, $t^*$ is the smallest $t$ in
  $\linterval{0}{D_n-J_n}$ satisfying $t = \rbf(t)$. Thus, $t^*$ is a solution
  for Problem~\eqref{opt:rta}.
\end{proof}

Similar IP or MILP (mixed-integer linear programming) formulations may be found
in the literature~\citep[see, for
instance,][]{baruahRealtimeSchedulingSporadic2005,
  zengEfficientFormulationRealTime2013}.\footnote{The formulation proposed
  (informally) by \citet{baruahRealtimeSchedulingSporadic2005} is missing the
  constraint $t > 0$. Thus, it is trivially satisfied by setting all variables
  to $0$.} In most cases in the literature, new methods for solving the
formulations are not proposed: the use of IP/MILP solvers or generic methods
such as branch-and-bound for solving IPs is implicitly suggested. We will
propose a new cutting-plane algorithm for solving IP-FP in
Section~\ref{sec:cp_kern}. \citet{zengEfficientFormulationRealTime2013} claim
that since their formulation contains $\bigO(n^2)$ variables ``solving the MILP
problem in feasible time is impossible for medium and sometimes small size
systems.'' This claim, however, seems to be based on intuition rather than
evidence. Their formulation can be trivially decomposed into $n$ IPs each
containing $\bigO(n)$ variables and constraints, and we will show soon that our
cutting-plane algorithm can solve IP-FP in fewer iterations than RTA uses to
solve the corresponding instance of Problem~\eqref{opt:rta}.

\subsection{IP-EDF}%
\label{sec:ip_edf}%

The EDF schedulability condition is more complex than the FP schedulability
condition, and the process of formulating it as an integer program is
error-prone (see discussion in Appendix~\ref{sec:A3}). Instead of trying to
formulate the EDF schedulability condition as a single IP, we divide it into
$\bigO(n)$ subproblems with simple IP formulations. Recall our assumption that
in EDF systems tasks are listed in a nondecreasing order using the key
$\hat{D}_j - T_j$ for each task $j$. Now, we consider the following $n$
subintervals of $\rinterval{\hat{D}_{\min}}{L}$:
\begin{align*}
  &\rinterval{\hat{D}_{1} - T_1}{\hat{D}_2 - T_2} \cap \rinterval{\hat{D}_{\min}}{L},\\
  &\rinterval{\hat{D}_2 - T_2}{\hat{D}_3 - T_3} \cap \rinterval{\hat{D}_{\min}}{L},\\
  &\dotsc,\\
  &\rinterval{\hat{D}_{n-1} - T_{n-1}}{\hat{D}_n - T_n} \cap \rinterval{\hat{D}_{\min}}{L},\\
  &\rinterval{\hat{D}_n - T_n}{L} \cap \rinterval{\hat{D}_{\min}}{L}.
\end{align*}
We ignore $\rinterval{\hat{D}_{\min}}{\hat{D}_1 - T_1}$ because for any $m$ such
that $\hat{D}_m = \hat{D}_{\min}$, we must have
\[
  \hat{D}_1 - T_1 \le \hat{D}_m - T_m < D_m = \hat{D}_{\min},
\]
and thus the interval is empty. Thus, the family of subintervals is a partition
of $\rinterval{\hat{D}_{\min}}{L}$.

Some or all of the intervals considered in the above list may be empty; for
instance, only the last interval is nonempty in a constrained-deadline system.
Thus, in the $n$ intervals listed above, we can restrict our attention from the
$p$-th interval to the $q$-th interval, where
\begin{align*}
  p &= \min \{i \in [n] \mid i == n \lor \hat{D}_{i+1} - T_{i+1} \ge \hat{D}_{\min}\}\\
  q &= \max \{i \in [n] \mid \hat{D}_{i} - T_{i} < L\}
\end{align*}

For any $k \in [q] \setminus [p-1]$, let $\rinterval{a_k}{b_k}$ refer to the
$k$-th interval. The $\dbf$ function in $\rinterval{a_k}{b_k}$ can be simplified
to a simpler function $\dbf_k$ given by
\[
  t \mapsto \sum_{i \in [k]} \left\lfloor\frac{t + T_i -
      \hat{D}_i}{T_i}\right\rfloor C_i
\]
because
\[
  t < b_k \le \hat{D}_{k+1} - T_{k+1} \le \hat{D}_{k+2} - T_{k+2} \le \dotsb \le
  L.
\]
Using $\dbf_k$, the EDF unschedulability condition (see
Condition~\eqref{cond:edf}) can be rewritten as follows: an EDF system is
unschedulable if and only if $\sum_{j \in [n]} U_j > 1$ or there exists a $k \in
[n]$ such that
\begin{equation}\label{cond:edf_2}
  \exists t \in \rinterval{a_k}{b_k}: t < \dbf_k(t)
\end{equation}
By replacing $t$ with $-t$, Condition~\eqref{cond:edf_2} can be rewritten as
\[
  \exists t \in \linterval{-b_k}{-a_k}: t + \dbf_k(-t) \ge 1
\]
Using the identity $\floor{-x} = -\ceil{x}$, we have
\begin{align*}
  \dbf_k(-t) &= \sum_{i \in [k]} \left\lfloor\frac{-t + T_i - \hat{D}_i}{T_i}\right\rfloor C_i\\
             &= -\sum_{i \in [k]} \left\lceil\frac{t + \hat{D}_i - T_i}{T_i}\right\rceil C_i.
\end{align*}
Thus, the condition can be rewritten again as
\begin{equation}%
  \label{cond:edf_3}%
  \exists t \in \linterval{-b_k}{-a_k}: \sum_{i \in [k]}
  \left\lceil\frac{t + \hat{D}_i - T_i}{T_i}\right\rceil C_i + 1 \le t.
\end{equation}
Condition~\eqref{cond:edf_3} is satisfied if and only if the problem
\begin{equation}%
  \label{opt:qpa_2}%
  \min \left\{t \in \linterval{-b_k}{-a_k} \cap \Z \mathrel{\Bigg\vert} \sum_{i \in [k]}
    \left\lceil\frac{t + \hat{D}_i - T_i}{T_i}\right\rceil C_i + 1 \le t\right\}
\end{equation}
has an optimal solution. The above analysis is distilled into
Algorithm~\ref{alg:branch}.

\begin{algorithm}[t]%
  \caption{Solve EDF schedulability by branching.}%
  \label{alg:branch}%
  \begin{algorithmic}[1]%
    \State{Sort tasks in a nondecreasing order using the key $i \mapsto \hat{D}_i-T_i$.}%
    \If{$\hat{D}_{\min} \ge L$}%
      \State\Return\textrm{``infeasible''}%
    \EndIf%
    \State{$p \gets 1$}%
    \While{$p \le n-1$}%
      \If{$\hat{D}_{p+1} - T_{p+1} \ge \hat{D}_{\min}$}%
        \State{\Break}%
      \EndIf%
      \State{$p \gets p+1$}%
    \EndWhile%
    \State{$q \gets n$}%
    \While{$q \ge 1$}%
      \If{$\hat{D}_{q} - T_{q} < L$}%
        \State{\Break}%
      \EndIf%
      \State{$q \gets q-1$}%
    \EndWhile%
    \State{$k \gets q$}%
    \While{$k \ge p$}%
      \State{$a_k \gets \max(\hat{D}_{\min}, \hat{D}_{k} - T_{k})$}%
      \If{$k = n$}%
        \State$b_k \gets L$
      \Else%
        \State$b_k \gets \hat{D}_{k+1} - T_{k+1}$%
      \EndIf%
      \State{Solve problem~\eqref{opt:qpa_2}.}%
      \If{an optimal solution exists}%
        \State\Return{the negation of optimal value}%
      \EndIf%
      \State$k \gets k - 1$%
    \EndWhile%
    \State\Return\textrm{``infeasible''}%
  \end{algorithmic}
\end{algorithm}

Since Problem~\eqref{opt:qpa_2} is very similar to Problem~\eqref{opt:rta}, it
can be formulated as an IP similar to IP-FP\@:
\begin{equation*}%
  \begin{array}{rrclcl}
    \min & \multicolumn{3}{l}{t} \\
    \textrm{s.t.}&t &\ge &-b_k+1\\
         &t &\le &-a_k\\
         &t - \sum_{j \in [k]} C_jx_j&\ge &1\\
         &T_ix_i - t&\ge & \hat{D}_i - T_i, &&i \in [k]\\
         &t &\in &\Z\\
         &x &\in &\Z^k
  \end{array}
\end{equation*}
This IP is equivalent to Problem~\eqref{opt:qpa_2}, and we call it IP-EDF\@. We
omit the proof because it is the same as the proof for Theorem~\ref{thm:ip_fp}.

\begin{thm}%
  \label{thm:ip_edf}%
  Problem~\eqref{opt:qpa_2} has a solution $t^*$ if and only if IP-EDF has an
  optimal solution $(t^*,x^*)$ for some $x^* \in \Z^k$.
\end{thm}

Thus, EDF schedulability can be solved by combining Algorithm~\ref{alg:branch}
and some method for solving IP-EDF\@. We demonstrate the process through an
example.

\begin{exmp}
  The EDF system in the Table~\ref{tab:ex2} misses a deadline at time $10$. The
  tasks are sorted in a nondecreasing order using the key $i \mapsto
  \hat{D}_i-T_i$. $\hat{D}_{\min}$ and $L$ equal $10$ and $13$ respectively.
  \begin{table}
    \centering%
    \begin{tabular}{c c c c}\toprule
      $i$ & $C_i$ & $D_i$ & $T_i$\\\midrule
      $1$ & $6$ & $10$ & $17$\\
      $2$ & $5$ & $10$ & $13$\\
      $3$ & $1$ & $31$ & $20$\\\bottomrule
    \end{tabular}%
    \caption{An EDF task system with arbitrary deadlines and zero release
      jitter.}%
    \label{tab:ex2}%
  \end{table}
  Running Algorithm~\ref{alg:branch} on the system, we find that the condition
  on line 2 is not satisfied. The calculation on lines 5--18 yields $p = 2$ and
  $q = 3$. In the loop, we solve Problem~\eqref{opt:qpa_2} for $(a_3,b_3) = (11,
  13)$ and $(a_2,b_2) = (10, 11)$ in order. In the former case, we do not find a
  feasible solution; in the latter case, we get the following IP\@:
  \begin{equation*}%
    \begin{array}{rrclcl}
      \min & \multicolumn{3}{l}{t} \\
      \textrm{s.t.}&t &\ge &-10\\
           &t &\le &-10\\
           &t - 6x_1 - 5x_2&\ge &1\\
           &17x_1 - t&\ge & -7\\
           &13x_2 - t&\ge & -3\\
           &t, x_1, x_2 &\in &\Z
    \end{array}
  \end{equation*}
  $(t,x_1,x_2) = (-10, -1, -1)$ is an optimal solution for the IP, and hence the
  system is unschedulable.
\end{exmp}

\section{The Kernel}%
\label{sec:kern}%

The following problem is called \emph{the kernel}:
\begin{equation}%
  \label{opt:abs}%
  \min \left\{t \in \interval{a}{b} \cap \Z \mathrel{\Bigg\vert}
    \sum_{j \in [n]} \ceil[\Bigg]{\frac{t + \alpha_j}{T_j}}C_j + \beta \le t\right\},
\end{equation}
where
\begin{itemize}
\item $n$ is a nonnegative integer.
\item For any $i \in [n]$, $C_i,T_i \in \Z_{> 0}$ are constants, and $U_i$ denotes
  the ratio $C_i / T_i$. We assume that $\sum_{j \in [n]} U_j \le 1$.
\item $\alpha \in \Z^{n}$, $\beta \in \Z$, $a \in \Z$, and $b \in \Z$ are
  constants.
\end{itemize}
It may be useful to view $C_i$, $T_i$, and $U_i$ as the wcet, period, and
utilization, respectively, of task $i$; on the other hand, this interpretation
can be ignored and $C$, $T$, and $U$ may be viewed simply as vectors of
dimension $n$ that satisfy the above properties.

If $(n,\alpha,\beta,a,b)$ is equal to $(n, J, 0, 1, D_n-J_n)$, then the kernel
is identical to Problem~\eqref{opt:rta}. If $(n,\alpha,\beta,a,b)$ is equal
to $(k,\hat{D}-T, 1, -b_k+1, -a_k)$, then the kernel is identical to
Problem~\eqref{opt:qpa_2}. Thus, the FP and EDF schedulability problems
reduce to the kernel. If $(n,\alpha,\beta,a,b)$ is equal to $(n, J, 0, 1,
\hyp)$, then the kernel is identical to the problem of computing $L_a$ (see
Definition~\ref{def:La}). In the next subsections, we will propose two
algorithms to solve the kernel, FP-KERN and CP-KERN\@.

\subsection{FP-KERN}%
\label{sec:fp_kern}%

If an optimal solution $t^*$ for the kernel exists and does not trivially equal
$a$, then $t^*$ is the smallest $t \in \interval{a}{b}$ that satisfies
\[
  \sum_{j \in [n]} \ceil[\Bigg]{\frac{t + \alpha_j}{T_j}}C_j + \beta = t;
\]
equivalently, $t^*$ is the smallest fixed point of $\phi$ in $\interval{a}{b}$
where $\phi$ is given by
\[
  t \mapsto \sum_{j \in [n]} \ceil[\Bigg]{\frac{t + \alpha_j}{T_j}}C_j + \beta.
\]
Since $\phi$ is a monotonically nondecreasing step function. The number of steps
in $\interval{a}{b}$ is roughly equal to
\[
  \sum_{j \in [n]} \frac{b-a}{T_i} = \frac{(b-a)n}{\Thm}
\]
From Theorem~\ref{thm:fixed}, $t^*$ can be found by fixed-point iteration in no
more than $n(b-a) / \Thm$ steps. Since each iteration requires $\bigO(n)$
arithmetic operations, the algorithm requires $\bigO((b-a)n^2 / \Thm)$ time. The
fixed-point iteration algorithm is called FP-KERN and is listed as
Algorithm~\ref{alg:fp}.

\begin{algorithm}[t]%
  \caption{FP-KERN, a fixed-point iteration algorithm for the kernel.}%
  \label{alg:fp}%
  \begin{algorithmic}[1]%
    \If{$a > b$}%
      \State\Return\textrm{``infeasible''}%
    \EndIf%
    \If{$n = 0$}%
      \State\Return$a$%
    \EndIf%
    \State{}Initialize $\u{t}$ to $a$.%
    \Repeat%
      \State$v \gets \phi(\u{t})$
      \If{$v \le \u{t}$}%
        \State\Return$\u{t}$%
      \EndIf%
      \State$\u{t} \gets v$%
    \Until$\u{t} > b$%
    \State\Return\textrm{``infeasible''}
  \end{algorithmic}
\end{algorithm}

\begin{thm}%
  \label{thm:fp-kern}%
  FP-KERN (Algorithm~\ref{alg:fp}) is a fixed-point iteration algorithm for the
  kernel with
  \[
    \bigO\left(\frac{(b-a)n^2}{\Thm}\right)
  \]
  running time.
\end{thm}

\subsection{IP-KERN}%
\label{sec:ip_kern}%

Consider the following IP\@:
\begin{equation}%
  \label{ip:kern0}%
  \begin{array}{rrclcl}
    \min & \multicolumn{3}{l}{t} \\
    \textrm{s.t.}&t &\ge &a\\
         &t &\le &b\\
         &t - \sum_{j \in [n]} C_jx_j&\ge &\beta\\
         &T_ix_i - t&\ge &\alpha_i, &&i \in [n]\\
         &t &\in &\Z\\
         &x &\in &\Z^n
  \end{array}
\end{equation}

We can compose this IP with the reductions discussed near the start of
Section~\ref{sec:kern} to get IP-FP and IP-EDF\@. The proof for the validity of
the formulation IP-FP (Theorem~\ref{thm:ip_fp}) can be trivially modified to
show that IP~\eqref{ip:kern0} is valid for the kernel. Although
IP~\eqref{ip:kern0} is a perfectly good formulation, we also include lower
bounds for $x$ in our final IP formulation for the kernel. From $t \ge a$ and
the fourth constraint in IP~\eqref{ip:kern0}, we can infer that
\[
  x_i \ge \frac{a + \alpha_i}{T_i}, \quad i \in [n].
\]
Since $x_i$ is integral, we can infer stronger lower bounds:
\[
  x_i \ge \ceil[\Bigg]{\frac{a + \alpha_i}{T_i}}, \quad i \in [n].
\]
Adding these lower bounds to the IP, we have our final IP formulation for the
kernel, which we call IP-KERN\@:
\begin{equation}%
  \label{ip:kern}%
  \begin{array}{rrclcl}
    \min & \multicolumn{3}{l}{t} \\
    \textrm{s.t.}&t &\ge &a\\
         &t &\le &b\\
         &t - \sum_{j \in [n]} C_jx_j&\ge &\beta\\
         &T_ix_i - t&\ge &\alpha_i, &&i \in [n]\\
         &x_i &\ge &\u{x}_i, &&i \in [n]\\
         &t &\in &\Z\\
         &x &\in &\Z^n
  \end{array}
\end{equation}
Here, $\u{x}_i$ is a lower bound for $x_i$ given by
\begin{equation}%
  \label{eq:ux}%
  \u{x}_i = \ceil{(a + \alpha_i)/ T_i}.
\end{equation}
This equation holds initially but $\u{x}$ is repeatedly updated to larger
integral vectors in the cutting-plane algorithm described in the next section.
The next theorem follows from the above discussion.

\begin{thm}%
  \label{thm:ip_kern}%
  The kernel has a solution $t^*$ if and only if IP-KERN has an optimal solution
  $(t^*,x^*)$ for some $x^* \in \Z^k$.
\end{thm}

\subsection{A cutting-plane algorithm for IP-KERN}\label{sec:cp_kern}

\subsubsection{The relaxation}

When solving an IP using cutting planes, we must solve a relaxation of the IP\@.
Usually, a linear relaxation is used, but we will also drop the upper and lower
bounds on $t$ in our relaxation.

\begin{equation}%
  \label{lp:relax}%
  \begin{array}{rrclcl}
    \min & \multicolumn{3}{l}{t} \\
    \textrm{s.t.}&t - \sum_{j \in [n]} C_jx_j&\ge &\beta\\
         &T_ix_i - t&\ge &\alpha_i, &&i \in [n]\\
         &x_i &\ge &\u{x}_i, &&i \in [n]\\
         &t &\in &\R\\
         &x &\in &\R^n
  \end{array}
\end{equation}

Let $t^*$ denote the optimal value of this relaxation. Since the objective is to
minimize $t$, if the optimal value of the relaxation is greater than $b$, then
the proper linear relaxation of IP-KERN is infeasible, and hence IP-KERN is also
infeasible (see Section~\ref{sec:cp}). Thus, if $t^* > b$, then we can terminate
the algorithm, returning ``infeasible''. To understand why $t^* \ge a$ is
dropped from the IP, note that
\begin{align*}
  t^* &\ge \sum_{j \in [n]} C_jx_j + \beta\tag{using the first constraint in the LP}\\
      &\ge \sum_{j \in [n]} C_j\u{x}_j + \beta\tag{using the third constraint in the LP}\\
      &\ge \sum_{j \in [n]} C_j\ceil{(a + \alpha_j)/ T_j} + \beta\tag{using Equation~\ref{eq:ux}}
\end{align*}
Thus, if $t^* \le a$ then we must have
\[
  \sum_{j \in [n]} C_j\ceil{(a + \alpha_j)/ T_j} + \beta \le a
\]
which implies that $t = a, x_j = \ceil{(a + \alpha_j)/ T_j}$ is a feasible, and
hence optimal, solution to IP-KERN\@. Thus, if $t^* \le a$, then we can
terminate the algorithm returning $a$.

We will see shortly that the form of this LP remains constant for the entire
execution of the cutting-plane algorithm: only $\u{x}$ is updated to a larger
integral vector in each iteration.

\subsubsection{The separation problem}

Solving the separation problem is an essential step in any cutting-plane
algorithm. Let us assume that $(t^*,x^*)$ is an optimal solution for the
relaxation, $(t^{**}, x^{**})$ is an optimal solution for the IP, and $(t,x)$ is
a feasible solution for the IP\@. Then, $(t^*,x^*)$ must satisfy the following
equations:
\begin{align*}
  t^* &= \textstyle\sum_{j \in [n]} C_jx^*_j\\
  x^*_i &= \max\left\{\frac{t^*+\alpha_i}{T_i}, \u{x}_i\right\}, \quad i \in [n]
\end{align*}
If $x^*$ is integral, then $t^*$ is also integral, and we have found an optimal
solution to IP-KERN\@. Otherwise, there must exist a $j \in [n]$ such that
$x^*_j$ is fractional. $x^*_j$ cannot equal $\u{x}_j$, which is integral;
therefore, it must equal $(t^* + \alpha_j) / T_j$. Since $t^*$ is a dual bound
for the IP and the objective is to minimize $t$, we must have
\[
  t \ge t^{**} \ge t^*.
\]
Then, using the fourth constraint in the IP, $(t,x)$ must satisfy
\[
  x_j \ge \frac{t + \alpha_j}{T_j} \ge \frac{t^* + \alpha_j}{T_j} = x^*_j.
\]
Since $x_j$ is integral, the stronger inequality
\[
  x_j \ge \ceil{x^*_j}
\]
is also valid. Since $x^*_j$ is fractional, $x_j^* \ge \ceil{x^*_j}$ does not
hold. Thus, the inequality $x_j \ge \ceil{x^*_j}$ separates the current solution
$(t^*,x^*)$ from the convex hull of feasible solutions and is a cut (see
Section~\ref{sec:cp}). For each fractional element in $x^*_j$, a cut is added to
the IP by updating the lower bound $\u{x}_j$:
\[
  \u{x}_j \gets \ceil{x^*_j}.
\]
This step requires $\bigO(n)$ arithmetic operations, it does not change the
number of variables or the number of constraints in the IP, and it maintains the
integrality of $\u{x}$, which is utilized in generating cuts in the next
iteration.

Note that we could have used general cuts like Gomory cuts~\citep[see, for
instance,][Ch. 8.6]{alma991020577069705251} in our cutting-plane algorithm but
we choose to use a cut generation technique that is specialized for our problem.
While we do not claim to generate the strongest cuts, our cuts have nice
properties: they are generated efficiently and do not change the number of
variables and number of constraints in the IP\@. Moreover, we will show in the
next section that FP-KERN is also a cutting-plane algorithm with the same cut
generation technique but a different relaxation.

The cutting-plane algorithm that we have developed in this section is called
CP-KERN and is listed as Algorithm~\ref{alg:cp}. For any $i \in [n]$, the
minimum (resp., maximum) value of $\u{x}_i$ is $\ceil{(a + \alpha_i) / T_i}$
(resp., $\ceil{(b + \alpha_i) / T_i}$); thus, $\u{x}_i$ can assume
$\bigO((b-a)/T_i)$ values. Since at least one element in $\u{x}$ is incremented
in each iteration, the maximum number of iterations is roughly equal to
\[
  \sum_{j \in [n]} \frac{b-a}{T_i} = \frac{(b-a)n}{\Thm}
\]
In each iteration, the linear program can be solved in time polynomial in the
size of the representation of the
program~\citep{khachiyanPolynomialAlgorithmsLinear1980}. The size of the
representation of the program is linearly bounded by the size of the
representation of the original instance of the kernel, denoted $\lvert I
\rvert$. Thus, the algorithm has $\bigO((b-a)\poly(\lvert I \rvert) / \Thm)$
running time.

\begin{algorithm}[t]%
  \caption{CP-KERN, a cutting-plane algorithm for the kernel.}%
  \label{alg:cp}%
  \begin{algorithmic}[1]%
    \If{$a > b$}%
      \State\Return\textrm{``infeasible''}%
    \EndIf%
    \If{$n = 0$}%
      \State\Return$a$%
    \EndIf%
    \parState{Initialize $\u{x}_i = \ceil{(a + \alpha_i) / T_i}$ for
      all $i \in [n]$.}%
    \Repeat%
      \parState{Solve LP~\eqref{lp:relax}. If infeasible, return
        \textrm{``infeasible''}, else let $(t^*,x^*)$ be the optimal solution.}%
      \If{$t^* \le a$}%
        \State{}\Return$a$%
      \EndIf%
      \If{$t^* > b$}%
        \State{}\Return\textrm{``infeasible''}%
      \EndIf%
      \parState{Update $\u{x}_i$ to $\ceil{x^*_i}$ for all $i \in
        [n]$.}%
    \Until$\u{x}$ stabilizes.%
    \State\Return$C\cdot \u{x} + \beta$%
  \end{algorithmic}
\end{algorithm}

\begin{thm}%
  \label{thm:cp}%
  CP-KERN (Algorithm~\ref{alg:cp}) is a cutting-plane algorithm for the kernel
  with
  \[
    \bigO\left(\frac{(b-a)\poly(\lvert I \rvert)}{\Thm}\right)
  \]
  running time.
\end{thm}

\addtocounter{exmp}{-1}%
\begin{exmp}[contd.]
  We will solve the schedulability problem for the lowest priority task in
  the FP task system in Table~\ref{tab:ex1} by reducing it to the kernel,
  formulating the kernel as IP-KERN, and solving IP-KERN by using CP-KERN\@. The
  kernel is given by
  \[
    \min \left\{t \in \interval{1}{D_3} \cap \Z \mathrel{\Bigg\vert} \sum_{j \in
        [3]} \ceil[\Bigg]{\frac{t}{T_i}}C_j \le t\right\}.
  \]
  IP-KERN is given by
  \[
    \begin{array}{rrclcl}
      \min & \multicolumn{3}{l}{t} \\
      \textrm{s.t.}&t &\ge &1\\
           &t &\le &150\\
           &t - 20x_1 - 10x_2 - 33x_3&\ge &0\\
           &40x_1 - t&\ge &0\\
           &50x_2 - t&\ge &0\\
           &150x_3 - t&\ge &0\\
           &x_i&\ge &1, &&i \in [3]\\
           &t, x_1, x_2, x_3 &\in &\Z
    \end{array}
  \]
  We describe the execution of CP-KERN for the above IP\@. Initially we have
  $\u{x}_i = 1$ for all $i \in [3]$. The relaxation is given by
  \[
    \begin{array}{rrclcl}
      \min & \multicolumn{3}{l}{t} \\
      \textrm{s.t.}&t - 20x_1 - 10x_2 - 33x_3&\ge &0\\
           &40x_1 - t&\ge &0\\
           &50x_2 - t&\ge &0\\
           &150x_3 - t&\ge &0\\
           &x_i&\ge &\u{x}_i, &&i \in [3]\\
           &t, x_1, x_2, x_3 &\in &\R
    \end{array}
  \]
  We solve the relaxation and get the optimal solution $(t^*,x^*) = (110,
  {2.75,2.2,1})$. We update $\u{x}$ to $(3,3,1)$. We solve the new relaxation
  and get $(t^*,x^*) = (126, {3.15,3,1})$. We update $\u{x}$ to $(4,3,1)$. We
  solve the new relaxation and get $(t^*,x^*) = (143, {4,3,1})$. Since $\u{x}$
  has stabilized, we terminate the algorithm. Thus, Algorithm~\ref{alg:cp}
  solves the problem in $3$ iterations; recall from Figure~\ref{fig:rta} that
  RTA solved the problem in $5$ iterations.
\end{exmp}

\subsection{Comparing FP-KERN and CP-KERN}%
\label{sec:comp}%

Let $\mathcal{F}$ denote the family of cutting-plane algorithms for IP KERN
where each algorithm satisfies the following properties:
\begin{enumerate}
\item In each iteration, the algorithm solves a relaxation of IP-KERN where all
  integral variables are made continuous and some constraints are dropped
  optionally.
\item Cuts are generated by rounding up the values of the optimal $x$ values
  for the relaxation.
\end{enumerate}
Clearly, CP-KERN belongs to $\mathcal{F}$. Although the relaxation used by
CP-KERN drops the upper and lower bounds for $t$, these changes are addressed
immediately after the relaxation is solved (see discussion near the start of
Section~\ref{sec:cp_kern}). Thus, CP-KERN effectively solves the proper linear
relaxation of IP-KERN, and has the optimal convergence rate in $\mathcal{F}$. We
will show that FP-KERN also belongs to $\mathcal{F}$, but it has a suboptimal
convergence rate.

Algorithm~\ref{alg:fp} and Algorithm~\ref{alg:fp_2} are equivalent ways to
express FP-KERN\@. While Algorithm~\ref{alg:fp} uses two variables ($\u{t},v$),
Algorithm~\ref{alg:fp_2} uses $n+1$ variables ($\u{x}_1,\dotsc,\u{x}_n,v$). $v$
is a temporary variable in both cases; $\u{t}$ and $\u{x}_i$ denote lower bounds
for $t$ and $\ceil{(t + \alpha_i)/T_i}$ for the kernel. It is easy to see that
Algorithm~\ref{alg:fp} and Algorithm~\ref{alg:fp_2} are essentially the same:
the computation of $\phi$ in Algorithm~\ref{alg:fp} is broken into two steps:
the ceiling expressions are evaluated at the end of the loop and the resulting
values are combined with $C$ and $\beta$ to get $\phi(v)$ at the start of the
next loop.

\begin{algorithm}[t]%
  \caption{An algorithm equivalent to FP-KERN.}%
  \label{alg:fp_2}%
  \begin{algorithmic}[1]%
    \If{$a > b$}%
      \State\Return\textrm{``infeasible''}%
    \EndIf%
    \If{$n = 0$}%
      \State\Return$a$%
    \EndIf%
    \parState{Initialize $\u{x}_i = \ceil{(a + \alpha_i) / T_i}$ for
      all $i \in [n]$.}%
    \Repeat%
      \State$v \gets C\cdot\u{x} + \beta$
      \If{$v \le a$}%
        \State{}\Return$a$%
      \EndIf%
      \If{$v > b$}%
        \State{}\Return\textrm{``infeasible''}%
      \EndIf%
      \parState{Update $\u{x}_i$ to $\ceil{(v + \alpha_i) / T_i}$ for all $i \in [n]$.}%
    \Until$\u{x}$ stabilizes.%
    \State\Return$C\cdot \u{x} + \beta$%
  \end{algorithmic}
\end{algorithm}

From the description of FP-KERN in Algorithm~\ref{alg:fp_2}, we can see that
FP-KERN is almost the same as CP-KERN except for the fact that the latter solves
a relaxation of IP-KERN at the beginning of the loop. It turns out that FP-KERN
also solves a relaxation of IP-KERN at the beginning of the loop, and this
relaxation has the optimal value $C\cdot\u{x} + \beta$. The relaxation used by
FP-KERN makes all variables continuous and drops the upper and lower bounds of
$t$ and the constraint $Tx_i - t \ge \alpha_i$ for all $i \in [n]$; thus, the
relaxation used by FP-KERN is given by
\[
  \begin{array}{rrclcl}
    \min & \multicolumn{3}{l}{t} \\
    \textrm{s.t.} &t  - \sum_{j \in [n]} C_jx_j&\ge &\beta\\
         &x_i &\ge &\u{x}_i,&& i \in [n]\\
         &t &\in &\R\\
         &x &\in &\R^n
  \end{array}
\]
The optimal value for the relaxation simply equals $C\cdot\u{x} + \beta$. Since
the above relaxation drops more inequalities than the relaxation used by
CP-KERN, it does not find the proper dual bounds in each iteration. The
suboptimality of FP-KERN is corroborated by the example at the end of
Section~\ref{sec:cp_kern} and the empirical evaluations in Section~\ref{sec:emp}.

\begin{thm}\label{thm:cp_opt}
  CP-KERN (resp., FP-KERN) has an optimal (resp., suboptimal) convergence rate
  in the family $\mathcal{F}$ of cutting-plane algorithms for solving IP-KERN\@.
\end{thm}

The relaxation used by FP-KERN can be solved in $\Theta(n)$ time since it
reduces to evaluating the expression $C\cdot\u{x} + \beta$. In contrast, the
relaxation used by CP-KERN can be solved in $\poly(\lvert I \rvert)$ time by
using an algorithm such as the ellipsoid algorithm for solving general LPs. In
the next section, we propose a more efficient method for solving the relaxation
used by CP-KERN\@.

\subsection{A specialized method for solving the relaxation}%
\label{sec:eff}%

In this section, we develop a specialized algorithm for solving the dual of
LP~\eqref{lp:relax}. Recall from Section~\ref{sec:duality} that solving the dual
is effectively the same as solving the relaxation itself, except that the
notions of unboundedness and infeasibility are inverted. The dual can be
constructed by following a dualization recipe that may be found in any textbook
on linear programming~\citep[][pg.\ 85]{matousekUnderstandingUsingLinear2007};
therefore, we skip the details of the construction. Let $v,w,z$ be the variables
in the dual LP such that they correspond, in order, to the three inequalities in
the primal LP, i.e., LP~\eqref{lp:relax}. Then, the dual is given by:
\begin{equation}%
  \label{lp:dual_orig}%
  \begin{array}{rrclcl}
    \max &\multicolumn{3}{l}{\beta{}v + \alpha\cdot w + \u{x}\cdot z}\\
    \textrm{s.t.} &v - \sum_{j \in [n]} w_j &= &1\\
         &-C_iv + T_iw_i + z_i &= &0, &&i \in [n]\\
         &v &\in &\R_{\ge 0}\\
         &w,z &\in &\R^n_{\ge 0}
  \end{array}
\end{equation}
By substituting $z_i / T_i$ for $z_i$ for all $i \in [n]$ and eliminating $w$,
we get
\begin{equation}%
  \label{lp:dual}%
  \begin{array}{rrclcl}
    \max &\multicolumn{3}{l}{(\beta + U\cdot{}\alpha{})v + \sum_{j \in [n]} (T_j\u{x}_j - \alpha_j)z_j}\\
    \textrm{s.t.} &(1 - \sum_{j \in [n]}U_j)v + \sum_{j \in [n]} z_j &= &1\\
         &U_iv - z_i &\ge &0, &&i \in [n]\\
         &v &\in &\R_{\ge 0}\\
         &z &\in &\R^n_{\ge 0}
  \end{array}
\end{equation}

For a fixed $v$, LP~\eqref{lp:dual} is identical to a fractional knapsack
problem with $n$ items where $z_i$ is the weight of item $i$ in the knapsack,
$U_iv$ is the total weight of item $i$ available to the thief, $T_i\u{x}_i -
\alpha_i$ is the value of the item per unit weight, and the knapsack weighs
exactly $1 - (1 - \sum_{j \in [n]}U_j)v$. The fractional knapsack problem admits
a greedy solution in which the thief fills the knapsack with items in
nonincreasing order of value until the knapsack is full. By assuming that the
vectors $\u{x}$, $U$, $T$, and $\alpha$ are sorted in nonincreasing order of
value, i.e., $T_i\u{x}_i - \alpha_i$, we can infer that an optimal solution
exists such that for some $k \in [n]$,
\begin{align*}
  z_j &= U_jv, &&j \in [k-1]\\
  z_k &\le U_kv\\
  z_j &= 0, &&j \in [n] \setminus [k]
\end{align*}
For each $k \in [n]$, we can create a more restricted version of linear
program~\eqref{lp:dual} by adding the above constraints:
\begin{equation}%
  \label{lp:sub}%
  \begin{array}{rrclcl}
    \max &\multicolumn{3}{l}{(\beta + \sum_{j \in [n] \setminus [k-1]} U_j\alpha_j + \sum_{j \in [k-1]}
           C_j\u{x}_j)v +}\\
         &\multicolumn{3}{l}{(T_k\u{x}_k - \alpha_k)z_k}\\
    \textrm{s.t.} &(1 - \sum_{j \in [n] \setminus [k-1]}U_j)v + z_k &= &1\\
         &U_kv - z_k &\ge &0\\
         &v,z_k &\in &\R_{\ge 0}
  \end{array}
\end{equation}

The relationship between LP~\eqref{lp:dual} and LP~\eqref{lp:sub} is summarized
in the next lemma.

\begin{lem}%
  \label{lem:sub}%
  The following statements are true:
  \begin{enumerate}
  \item LP~\eqref{lp:dual} has an optimal solution $(v,z)$ if and only if for
    some $k \in [n]$ LP~\eqref{lp:sub} has an optimal solution $(v,z_k)$.
  \item LP~\eqref{lp:dual} is infeasible if and only if LP~\eqref{lp:sub} is
    infeasible for all $k \in [n]$.
  \item LP~\eqref{lp:dual} is unbounded if and only if LP~\eqref{lp:sub} is
    unbounded for some $k \in [n]$.
  \end{enumerate}
\end{lem}

In the remainder of this section, we analyze LP~\eqref{lp:sub} under different
conditions to find an algorithm for solving LP~\eqref{lp:dual}.

\begin{lem}\label{lem:lp_sub_sp1}
  If $\sum_{j \in [n]} U_j = 1$ and $\beta + U\cdot\alpha > 0$, then IP-KERN is
  infeasible.
\end{lem}
\begin{proof}[Proof Sketch]
  Let us assume that $\sum_{j \in [n]} U_j = 1$ and $\beta + U\cdot\alpha$ is
  positive. For $k = 1$ in LP~\eqref{lp:sub}, $v$ disappears from the first
  constraint and its coefficient in the objective function is positive; thus,
  the LP is unbounded. Then, using Lemma~\ref{lem:sub}, LP~\eqref{lp:dual} is
  also unbounded. Since LP~\eqref{lp:relax} is a dual of LP~\eqref{lp:dual}, it
  must be infeasible (see Section~\eqref{sec:duality}). Since
  LP~\eqref{lp:relax} is a relaxation of IP-KERN, the infeasibility of the LP
  implies the infeasibility of IP-KERN\@.
\end{proof}

We introduce a map $f: \{0\} \cup [n] \to \Q$ given by\footnote{A function like
  $f$ is used by Lu et al.\ in a heuristic method to solve FP
  schedulability~\citep{luPreciseSchedulabilityTest2006}. However, their method,
  being heuristic, tries to guess a good $k \in [n]$ and tries $f(k)$ as the
  next dual bound. If the guess turns out to be lower than the current dual
  bound, they backtrack to using RTA\@. In contrast, we always find the best
  dual bound for the relaxation. Nguyen et al.\ also use a similar function in a
  linear-time algorithm for FP schedulability with harmonic
  periods~\citep{nguyenExactPolynomialTime2022}. In contrast, CP-KERN works for
  arbitrary periods and for both FP and EDF schedulability.}
\begin{equation}%
  \label{def:f}%
  k \mapsto \frac{\beta + \sum_{j \in [n] \setminus [k]} U_j\alpha_j + \sum_{j \in [k]}
    C_j\u{x}_j}{1 - \sum_{j \in [n] \setminus [k]}U_j}.
\end{equation}
If $k = 0$ and $\sum_{j \in [n]}U_j = 1$, then $f(k)$ is not well-defined; in
this case, we restrict the domain of $f$ to $[n]$. The next two lemmas show that
optimal values of LP~\eqref{lp:sub} always lie in the image of $f$.

\begin{lem}\label{lem:lp_sub_sp2}
  If $f(0)$ is not well-defined and $\beta + U\cdot\alpha \le 0$, then the
  optimal value for LP~\eqref{lp:sub} is $f(1)$.
\end{lem}
\begin{proof}[Proof Sketch]
  If $k=1$ and $\sum_{j \in [n]}U_j = 1$, then it may be verified that the
  optimal solution is $(v,z_1) = (1/U_1,1)$ and the optimal value is
  \begin{align*}
    &(\beta + \sum_{j \in [n]} U_j\alpha_j) (1 / U_1) + (T_1\u{x}_1 - \alpha_1)\\
    = &\frac{\beta + \sum_{j \in [n] \setminus [1]} U_j\alpha_j +
        C_1\u{x}_1}{U_1}\\
    = &f(1)\tag{using Definition~\eqref{def:f}}
  \end{align*}
\end{proof}

\begin{lem}\label{lem:lp_sub_sp3}
  For any $k \in [n]$, if $f(k-1)$ is well-defined, then the optimal value for
  LP~\eqref{lp:sub} is $f(k)$ or $f(k-1)$.
\end{lem}
\begin{proof}[Proof Sketch]
  If $f(k-1)$ is well-defined, then we must have $k > 1$ or $\sum_{j \in [n]}U_j
  < 1$. Then, $z_k$ can be eliminated from LP~\eqref{lp:sub} to get
  \begin{equation*}
    \begin{array}{rrclcl}
      \max &\multicolumn{3}{l}{(1 - \sum_{j \in [n]
             \setminus [k]} U_j)(f(k) -
             (T_k\u{x}_k - \alpha_k))v +
             T_k\u{x}_k - \alpha_k}\\
      \textrm{s.t.} &v &\in &\interval{(1 - \sum_{j \in [n] \setminus [k]}U_j)^{-1}}{(1 - \sum_{j \in [n] \setminus [k-1]}U_j)^{-1}}
    \end{array}
  \end{equation*}
  From the definition of the kernel, we have $\sum_{j \in [n]} U_j \le 1$. Thus,
  the left end of the interval is positive, and hence the constraint $v \ge 0$,
  which was present in program~\eqref{lp:sub}, is redundant and is not included
  above. From the assumption that $k > 1$ or $\sum_{j \in [n]}U_j < 1$, it
  follows that the right end of the interval is well-defined and greater than
  the left end of the interval. Thus, program~\eqref{lp:sub} is feasible and
  bounded. The optimal objective value for the program must be achieved at one
  of the ends of the interval for $v$, the end being determined by the sign of
  the coefficient of $v$ in the objective. Using Lemma~\ref{lem:f_id}, the
  objective can also be written as
  \[
    (1 - \textstyle\sum_{j \in [n] \setminus [k-1]} U_j)(f(k-1) - (T_k\u{x}_k -
    \alpha_k))v + T_k\u{x}_k - \alpha_k.
  \]
  Substituting the two ends of the interval for $v$ in the two expressions for
  the objective, we get that the optimal value is $f(k)$ or $f(k-1)$.
\end{proof}

We accumulate Lemmas~\ref{lem:lp_sub_sp1},~\ref{lem:lp_sub_sp2},
and~\ref{lem:lp_sub_sp3} into one theorem:

\begin{thm}\label{thm:lp_relax_opt}%
  If $\sum_{j \in [n]} U_j = 1$ and $\beta + U\cdot\alpha > 0$, then IP-KERN is
  infeasible; otherwise, LP~\eqref{lp:relax} has optimal value $\max f$
  (assuming that all vectors are sorted in a nonincreasing order using the key
  $T_j\u{x}_j - \alpha_j$ for each $j \in [n]$).
\end{thm}

A detailed version of CP-KERN, with the algorithm for solving the relaxation by
computing $\max f$ embedded in it, is listed as Algorithm~\ref{alg:cp_full}.

\begin{algorithm}[!htbp]%
  \caption{CP-KERN, detailed}%
  \label{alg:cp_full}%
  \begin{algorithmic}[1]%
    \If{$a > b$}%
      \State\Return\textrm{``infeasible''}%
    \EndIf%
    \If{$n = 0$}%
      \State\Return$a$%
    \EndIf%
    \State{$(p, q, r) \gets (\beta, 1, \beta)$}%
    \ForAll{$i \in [n]$}%
      \State{$\u{x}_i \gets \ceil{(a + \alpha_i) / T_i}$}
      \State{$y_i \gets T_i\u{x}_i - \alpha_i$}%
      \State{$\pi_i \gets i$}%
      \State$(p, q, r) \gets (p + U_i\alpha_i, q - U_i, r + C_i\u{x}_i)$%
    \EndFor%
    \If{$p > 0$ and $q = 0$}%
      \State\Return\textrm{``infeasible''}%
    \EndIf%
    \State{$i_0 \gets 1$ \textbf{if} $q = 0$ \textbf{else} $i_0 \gets 0$}%
    \parState{Sort $\pi$ in a nonincreasing order using the key $i \mapsto
      y_i$.}%
    \While{$\top$}%
      \State{$(p, q, i) \gets (r, 1, n)$}%
      \While{$i > i_0$}
        \State{$k \gets \pi_i$}%
        \If{$p \le q * y_k$}%
          \State{$t^* \gets p / q$}%
          \State\Break%
        \EndIf%
        \State{$(p, q, i) \gets (p - C_k\u{x}_k + U_k\alpha_k, q - U_k, i -
          1)$}%
      \EndWhile%
      \If{$i = i_0$}%
        \State{$t^* \gets p / q$}%
      \EndIf%
      \State{\textbf{if} $t^* \le a$ \textbf{then} \Return$a$; \textbf{if} $t^*
        > b$ \textbf{then} \Return{} ``infeasible''.}
      \If{$i = n$}%
        \State{\Return{$t^*$}}%
      \EndIf%
      \ForAll{$j \in [n] \setminus [i]$}%
        \State{$k \gets \pi_j$}%
        \State{$d \gets \ceil{(t^* + \alpha_k) / T_k} - \u{x}_k$}%
        \State{$(\u{x}_k, y_k, r) \gets (\u{x}_k + d, y_k + T_kd, r + C_kd)$}%
      \EndFor%
      \parState{Sort $\pi$ in the range $i+1$ to $n$ in a nonincreasing order
        using the key $j \mapsto y_j$.}%
      \parState{Merge the two sorted parts of $\pi$.}%
    \EndWhile%
  \end{algorithmic}
\end{algorithm}

\begin{thm}%
  \label{thm:cp_comb}%
  CP-KERN (Algorithm~\ref{alg:cp_full}) solves the kernel in
  \[
    \bigO\left(\frac{(b-a)n^2}{\Thm}\right)
  \]
  running time.
\end{thm}
\begin{proof}[Proof sketch]
  Algorithm~\ref{alg:cp_full} is essentially the same as Algorithm~\ref{alg:cp}:
  in each iteration, we solve LP~\eqref{lp:relax} and then update $\u{x}_i$ to
  $\ceil{x^*_i}$ for all $i \in [n]$. However, there are some differences
  between the two algorithms:
  \begin{itemize}
  \item We check a new infeasibility condition on line 14. Since $p = \beta +
    U\cdot\alpha$ and $q = 1 - \sum_{j \in [n]} U_j$ on this line, the condition
    is equivalent to the infeasibility condition in
    Theorem~\ref{thm:lp_relax_opt}.
  \item We maintain two vector variables $y$ and $\pi$, where $y_i = T_i\u{x}_i
    - \alpha_i$ for all $i \in [n]$, and $\pi$ stores the indices of $y$ so that
    for any $i \le j$ we have
    \[
      y_{\pi_i} \ge y_{\pi_j}.
    \]
    We also maintain a scalar variable $r = \beta + C\cdot\u{x} = f(n)$.
  \item We compute an optimal value $t^*$ for LP~\eqref{lp:relax} on lines
    20--31. On line 17, $i_0$ was set to the smallest point of the domain of
    $f$. We examine the values in the domain of $f$ excluding the smallest point
    $i_0$ in the inner while loop. On line 23, we have
    \[
      p/q = f(i)
    \]
    Thus, in the while loop, we search for an $i$ such that
    \[
      p / q \le y_k \iff f(i) \le T_{\pi_i}\u{x}_{\pi_i} - \alpha_{\pi_i} \iff
      f(i) \ge f(i-1).
    \]
    The second equivalence follows from Lemma~\ref{lem:f_id}. Thus, we are
    looking for a local maximum point of $f$ in its domain (excluding $i_0$).
    Using Corollary~\ref{cor:f_local_global}, any local maximum point is also a
    global maximum point. If we do not find any such point, then $i_0$ must be a
    global maximum point of $f$. We store the maximum value of $f$ in the
    variable $t^*$ on lines 24 and 30. If we have reached line 31, then $i$ is a
    global maximum point of $f$, and it corresponds to a solution for the
    LP~\eqref{lp:dual_orig} in which
    \begin{align*}
      &z_j > 0 \land w_j = 0, &&j \in \{\pi_1,\dotsc,\pi_{i}\}\\
      &z_j = 0 \land w_j > 0, &&j \in \{\pi_{i+1},\dotsc,\pi_{n}\}
    \end{align*}
    Using complementary slackness conditions (see Section~\ref{sec:duality}), we
    can infer that in the corresponding solution in LP~\eqref{lp:relax} we must
    have
    \begin{align}
      &x_j^* = \u{x}_j, &&j \in \{\pi_1,\dotsc,\pi_{i}\}\label{eq:x_opt_a}\\
      &x_j^* = (t^* + \alpha_j) / T_j, &&j \in \{\pi_{i+1},\dotsc,\pi_{n}\}\label{eq:x_opt_b}
    \end{align}
    We do not explicitly compute $x^*$ in our algorithm, but we use the above
    equations towards the end of the outer while loop when we update $\u{x}$.
  \item Instead of exiting the outer while loop when $\u{x}$ stabilizes, we exit
    the loop on line 34 if the condition on line 33 is satisfied. Recall that
    $i$ is a global maximum point of $f$ after line 31. Thus, the condition on
    line 33 is equivalent to $n$ being a global maximum point of $f$. From
    Equation~\eqref{eq:x_opt_a}, we must have $x_j^* = \u{x}_j$ for all $j \in
    [n]$. Furthermore, since the algorithm did not exit on 32, we have
    \[
      f(n) = t^* \in \interval{a}{b}.
    \]
    Thus, $(t^*,x^*) = f(n)$ is an optimal solution for IP-KERN\@.
  \item Instead of updating the full vector $\u{x}$, we do not update the
    entries corresponding to indices $\{\pi_1,\dotsc,\pi_{i}\}$ because they are
    already integral (using Equation~\eqref{eq:x_opt_a}). For any other index
    $j$, $\u{x}_j$ is updated to $\ceil{x^*_j} = \ceil{(t^* + \alpha_j) / T_j}$
    (using Equation~\eqref{eq:x_opt_b}). $r$ and $y$ are also modified to
    reflect the change in $\u{x}$. Since $\u{x}_j$ is unchanged for all $j \in
    \{\pi_1,\dotsc,\pi_{i}\}$, $\pi$ is already sorted in this range. Therefore,
    we sort $\pi_{i+1}, \dotsc, \pi_{n}$, and then we merge the two sorted parts
    of $\pi$ into one.
  \end{itemize}

  Let $n - m_i$ be the global maximum point of $f$ in the $i$-th iteration.
  Then, lines 20--40 run in $\bigO(m_i)$ time, and sorting on line 41 takes
  $\bigO(m_i\log(m_i))$ time. Since $m_i$ new values of $\u{x}$ are computed in
  the loop on line 36 and since each $\u{x}_j$ can assume $\bigO((b-a)/T_j)$
  values, over all iterations we must have
  \[
    \sum m_i = \bigO\left(\frac{\sum_{j \in [n]}(b-a)}{T_j}\right) =
    \bigO\left(\frac{(b-a)n}{\Thm}\right)
  \]
  Thus, the time taken by lines 20--41 over all iterations equals
  \begin{align*}
    &\bigO(\textstyle\sum m_i \log(m_i))\\
    = &\bigO(\textstyle\sum m_i \log(n)\tag{since $m_i \le n$})\\
    = &\bigO(\log(n)\textstyle\sum m_i)\\
    = &\bigO\left(\frac{(b-a)n\log(n)}{\Thm}\right)
  \end{align*}
  In one iteration, line 42 takes $\bigO(n)$ time. Since our analysis of the
  number of iterations in Algorithm~\ref{alg:cp} also applies to the current
  algorithm, over all iterations line 42 takes
  \[
    \bigO\left(\frac{(b-a)n^2}{\Thm}\right)
  \]
  time. Thus, in our analysis the work done by line 42 dominates the work done
  by the other lines. We think that the analysis of the number of iterations may
  be a little pessimistic and deserves further investigation.
\end{proof}

\subsection{Implications for FP and EDF schedulability}%
\label{sec:impl}%

Since Problem~\eqref{opt:rta} can be solved by reducing it to the kernel by
setting $(n,\alpha,\beta,a,b)$ to $(n, J, 0, 1, D_n-J_n)$, the following theorem
follows from Theorem~\ref{thm:cp_comb}:

\begin{thm}%
  \label{thm:fp}%
  Problem~\eqref{opt:rta} can be solved by CP-KERN in
  \[
    \bigO\left(\frac{(D_n-J_n)n^2}{\Thm}\right)
  \]
  running time.
\end{thm}

Since Problem~\eqref{opt:qpa_2} can be solved by reducing it to the kernel
by setting $(n,\alpha,\beta,a,b)$ to $(k,D-T-J, 1, -b_k+1, -a_k)$, the following
theorem follows from Theorem~\ref{thm:cp_comb}:

\begin{thm}%
  \label{thm:edf}%
  Problem~\eqref{opt:qpa_2} can be solved by CP-KERN in
  \[
    \bigO\left(\frac{(b_k-a_k)k^2}{\Thm^k}\right)
  \]
  running time, where $\Thm^k$ is the harmonic mean of the periods $T_1, \dotsc,
  T_k$.
\end{thm}

Since the problem of computing $L_a$ (see Definition~\ref{def:La}) can be solved
by reducing it to the kernel by setting $(n,\alpha,\beta,a,b)$ to
$(n,J,0,1,\hyp)$, the following theorem follows from Theorem~\ref{thm:cp_comb}:

\begin{thm}%
  \label{thm:La}%
  $L_a$ can be computed by CP-KERN in
  \[
    \bigO\left(\frac{\hyp{}n^2}{\Thm}\right)
  \]
  running time.
\end{thm}

In all cases, the worst-case running times of CP-KERN are the same as the
worst-case running times for the corresponding fixed-point iteration algorithms.
Since CP-KERN, unlike RTA and QPA, has an optimal convergence rate in the family
$\mathcal{F}$ (see Theorem~\ref{thm:cp_opt}), CP-KERN is a better algorithm than
RTA and QPA in theory.

\section{Empirical evaluation}%
\label{sec:emp}%

Our goal in this work is \emph{not} to determine the fastest schedulability test
amongst all known tests. As mentioned in the introduction, the schedulability
tests are, in general, incomparable and their running times depend on the sizes
of parameters like number of periods and $D_{\max} / T_{\min}$.\footnote{We
  often get contradictory evidence about the fastest schedulability tests from
  different groups of researchers when they measure the speed on random systems
  generated according to different
  criteria~\citep{biniSchedulabilityAnalysisPeriodic2004,
    davisEfficientExactSchedulability2008}.} However, the algorithms in the
family of cutting-plane algorithms considered in this work are comparable, and
we try to quantify the improvement in convergence rate and running time achieved
by our optimal cutting-plane algorithm CP-KERN over RTA and QPA for large
collections of synthetic systems. In keeping with the aims of the research, we
do not include all schedulability tests that exist in literature in our
empirical evaluations; instead, we restrict our attention to the family of
cutting-plane algorithms studied in this work.

\subsection{Environment}\label{sec:test_env}

We conducted the experiments on a MacBook Pro 2019 with a 2.6 GHz 6-Core Intel
Core i7 processor and 16 GB 2400 MHz DDR4 memory.

Our code for carrying out the above experiments is publicly
available~\citep{abhishek_singh_2023_7938966}. The code is written in mostly
written in Python and it contains
\begin{itemize}
\item methods for generating random systems of hard real-time tasks;
\item instrumented Python implementations of FP-KERN and CP-KERN that count the
  number of iterations used by the algorithms;
\item instrumented C++ implementations of FP-KERN and CP-KERN that measure the
  CPU time used by the algorithms;
\item reductions from FP and EDF schedulability to the kernel, including the
  branching approach for EDF described in Algorithm~\ref{alg:branch} and the
  optimizations discussed in Appendix~\ref{sec:A1};
\item tests for finding inconsistencies between CP-KERN and traditional
  schedulability tests (RTA and QPA); and
\item a script for running the experiments described above and for generating
  the associated data and images.
\end{itemize}

We depend on several Python packages including drs, numpy, matplotlib, pytest,
and scipy. Our C++ implementation of CP-KERN can produce incorrect answers due
to numerical issues stemming from the use of floating-point arithmetic and
integer overflow errors.\footnote{Using a commercial solver like CPLEX to solve
  IP-KERN can also produce incorrect answers due to numerical issues. Using a
  solver for exact IPs is recommended~\citep[see, for instance,
  ][]{cookHybridBranchandboundApproach2013, eiflerComputationalStatusUpdate2023}.}
In contrast, our Python implementation of CP-KERN uses rational arithmetic using
the fractions module and unlimited-precision integers. Thus, we think that the
Python implementation is more trustworthy. In all experiments, we check that the
results produced by FP-KERN and CP-KERN are consistent; thus, we have confidence
in the validity of the results derived from the C++ implementations.

\subsection{Data}\label{sec:test_data}

In each experiment, we randomly generate $10000$ synthetic task systems. Each
system contains $n$ tasks, where $n \in \{25, 50, 75\}$. We generate
utilizations so that their approximate sum is one of the values in $\{0.7, 0.8,
0.9\}$ by using the Dirichlet-Rescale (DRS)
algorithm~\citep{griffinGeneratingUtilizationVectors2020}. We sample wcets from
a loguniform distribution in $\interval{1}{1000}$, and then round each wcet up
to the nearest integer.\footnote{The fixed interval for the wcets reflects the
  idea that we do not expect a single reaction to an event to be too long in
  real systems. This idea is present in the \emph{synchrony hypothesis}, used by
  reactive languages like Esterel, which assumes that reactions are small enough
  to fit within a tick so that they appear to be instantaneous to the
  programmer~\cite{boussinotESTERELLanguage1991}.}. We compute period $T_i$ as
follows:
\[
  T_i \coloneq \ceil[\Bigg]{\frac{C_i}{U_i}}
\]
Since the periods are rounded up, the utilizations are rounded down; thus, the
utilization of the system is only approximately equal to our initially chosen
value. We generate FP systems and EDF systems differently.

\subsubsection{FP systems}

We generate FP systems with implicit deadlines and zero release jitter. Thus, we
set $D_i \coloneq T_i$ and $J_i \coloneq 0$ for all $i$.

In our current description, the parameter $D_n = T_n$ is generated just like the
other periods $T_1, \dotsc, T_{n-1}$. However, $D_n$ has a disproportionate
effect on the running times of algorithms for solving Problem~\eqref{opt:rta}.
In the current setting, Problem~\eqref{opt:rta} can be written as follows:
\[
  \min \left\{t \in \linterval{0}{D_n} \cap \Z \mathrel{\Bigg\vert} C_n +
    \sum_{j \in [n-1], T_j < D_n} \left\lceil\frac{t}{T_j}\right\rceil{}C_j \le
    t \right\}
\]
Thus, all the periods that are smaller than a randomly generated $D_n$ do not
contribute to the problem instance. Moreover, since $D_n$ is in the numerator in
the asymptotic characterization of worst-case running times in
Theorems~\ref{thm:rta} and~\ref{thm:fp}, a random instance will get solved
quickly if $D_n$ is small. Therefore, to reduce noise in the evaluation, we set
the parameters of task $n$ as follows:
\[
  D_n = T_n \coloneq 10^8, C_n \coloneq 100.
\]
This ensures that both RTA and CP-KERN are evaluated against challenging
instances of Problem~\eqref{opt:rta}. The high-priority subsystem $[n-1]$ is
generated randomly, as described above.

\subsubsection{EDF systems}

Recall from Section~\ref{sec:ip_edf} that we use a divide-and-conquer approach
to create $n$ subproblems where the $k$-th subproblem is the same as the
original problem except that $t$ is restricted to $\interval{a_k}{b_k} \subseteq
\rinterval{\hat{D}_{\min}}{L}$ (see Problem~\eqref{opt:qpa_2}); each subproblem
is then reduced to the kernel which is then solved by CP-KERN\@. The overall
algorithm has a branch-and-cut structure where the division into subproblems
corresponds to branching on the variable $t$ and using the cutting-plane
algorithm CP-KERN corresponds to the finding cuts on the nodes generated from
the branching.\footnote{Branch-and-cut is currently one of the most successful
  approaches for solving IPs; for more details see~\citep[][Sec.
  9.6]{alma991020577069705251}. A variant of QPA with some branching heuristics
  has also been proposed by \citet{zhangImprovementQuickProcessorDemand2009}.}
For constrained-deadline systems, only the branch
$\rinterval{\hat{D}_{\min}}{L}$ has feasible solutions. We restrict our
attention to constrained-deadline systems in the evaluation so that we do not
need to worry about the effect of branching strategies on the results.

We call the ratio $\delta_i = C_i / D_i$ the density of task $i$. We use the DRS
algorithm to generate the densities of the tasks so that their approximate sum
is one of the values in $\{1.25, 1.50, 1.75\}$. Since we are generating a
constrained-deadline system, we must ensure that the density of any task is at
least its utilization; so, we pass the utilization vector to the DRS algorithm
as a lower bound for the density vector. We compute the deadline $D_i$ as
follows:
\[
  D_i \coloneq \floor[\Bigg]{\frac{C_i}{\delta_i}}.
\]
Since deadlines are rounded down, the total density of the system is only
approximately equal to the initially chosen value.

\subsection{Experiment I}\label{sec:exp_1}

In this experiment, we compare the number of iterations used by RTA to the
number of iterations used by CP-KERN for synthetic instances of
Problem~\eqref{opt:rta}. For both algorithms, we use the initial value
\[
  \frac{C_n}{1 - \sum_{j \in [n-1]} U_j}.
\]
Initial values such as the one mentioned above for RTA have been discussed
extensively in the literature~\citep{sjodinImprovedResponsetimeAnalysis1998,
  brilInitialValuesOnline2003, davisEfficientExactSchedulability2008}; in
Appendix~\ref{sec:A1}, we show that the bound is implicitly used by CP-KERN and
can also be explicitly passed to CP-KERN by using an alternative reduction from
Problem~\eqref{opt:rta} to the kernel.

For $n = 25$ and $\sum_{j \in [n-1]} U_j \approx 0.9$, we show histograms of the
number of iterations of CP-KERN and RTA in Figure~\ref{fig:rta_1a}. The
histogram for CP-KERN is to the left of the histogram for RTA\@; thus, CP-KERN
has a better convergence rate than RTA\@. The histogram for CP-KERN is thinner
and taller than the histogram for RTA\@. Therefore, the convergence rate of
CP-KERN is more predictable than the convergence rate of RTA\@. We also show a
histogram of the ratio of the number of iterations of RTA to the number of
iterations of CP-KERN in Figure~\ref{fig:rta_1b}. The ratio is never below one,
which is in agreement with our theoretical result that CP-KERN is optimal with
respect to convergence rate, and the average ratio is about $2.6$.

\begin{figure}
  \centering%
  \begin{subfigure}[b]{0.75\linewidth}
    \centering%
    \includegraphics[width=\textwidth]{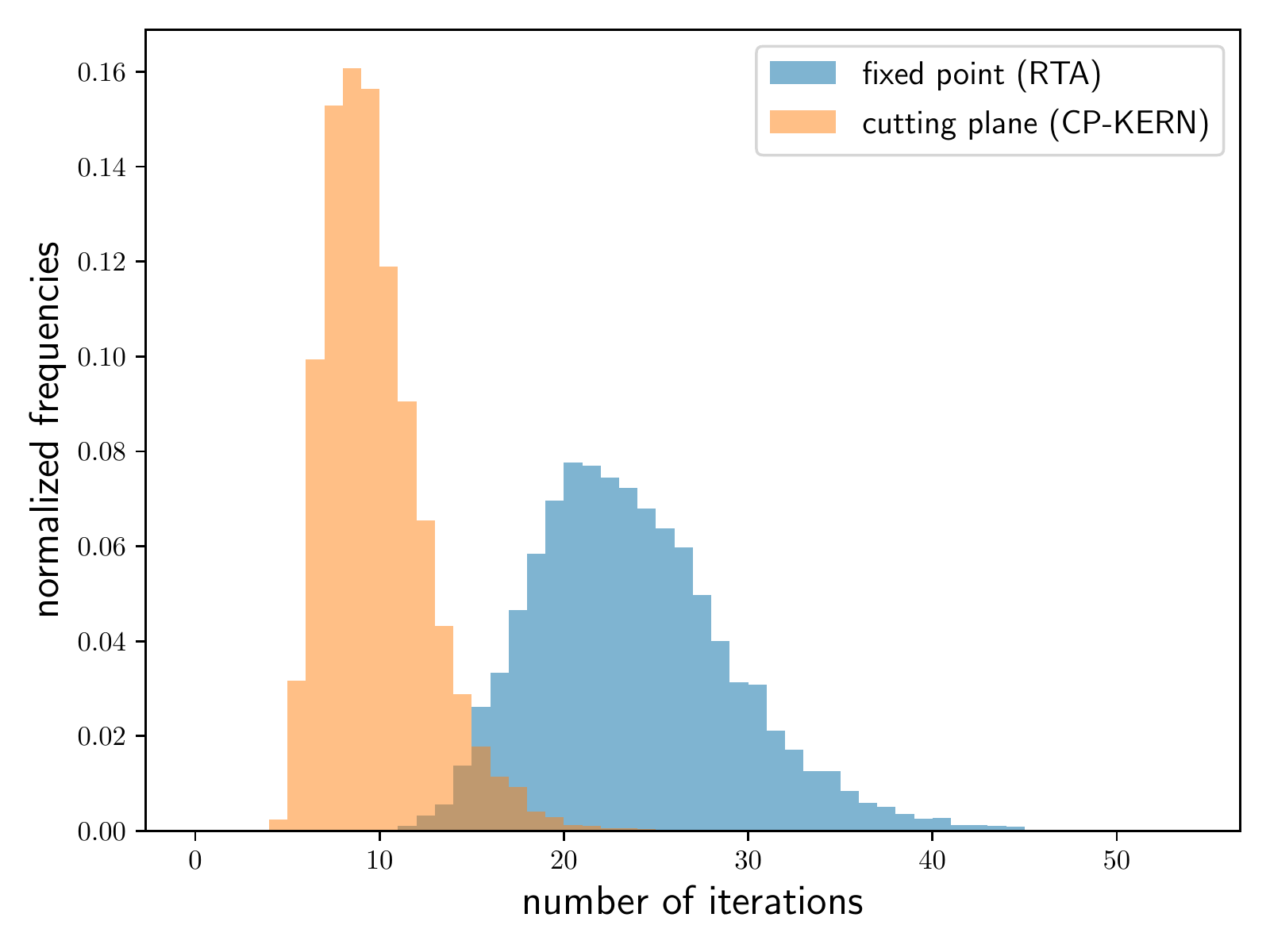}%
    \caption{}%
    \label{fig:rta_1a}%
  \end{subfigure}
  \begin{subfigure}[b]{0.75\linewidth}
    \centering%
    \includegraphics[width=\textwidth]{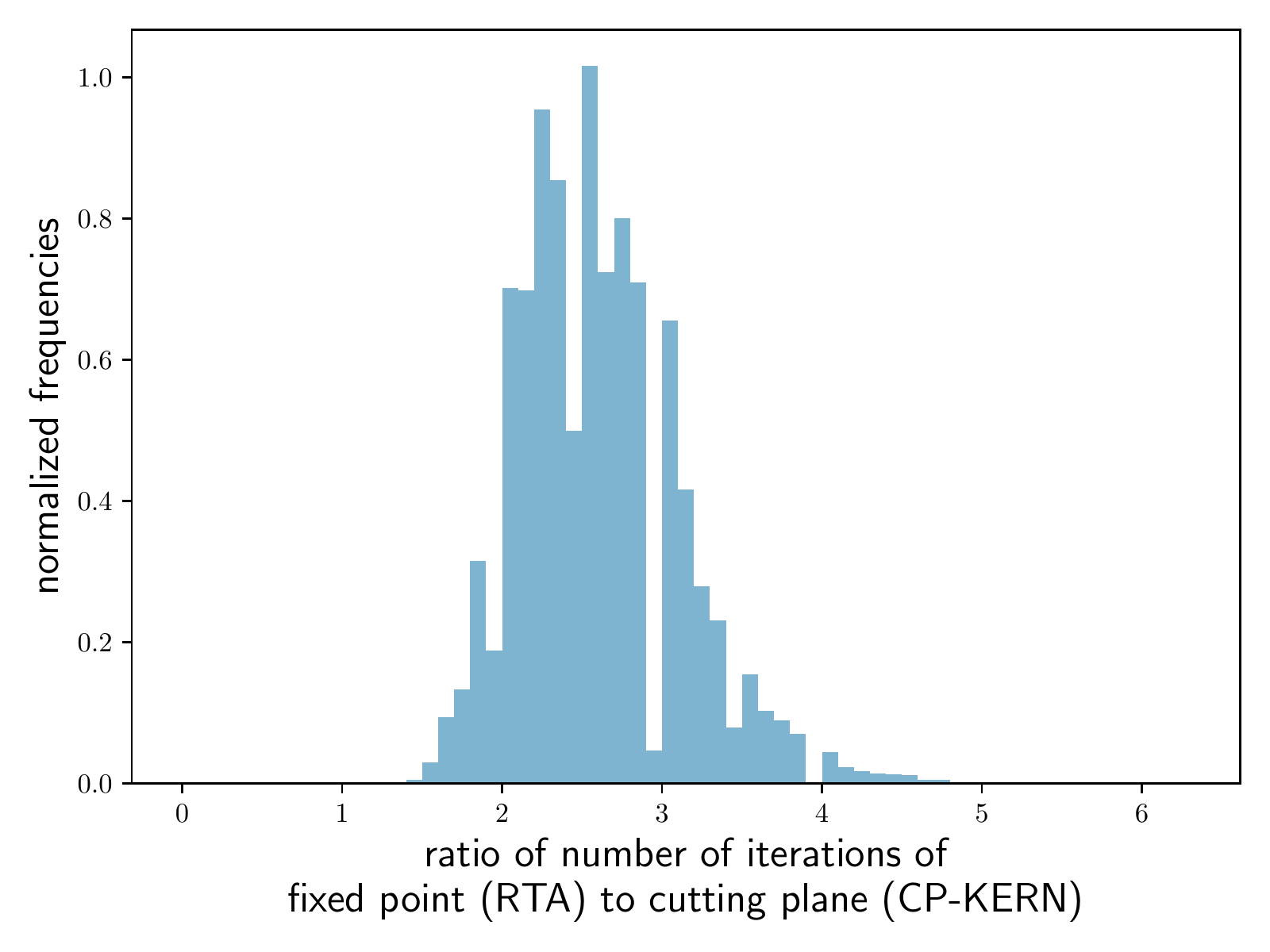}%
    \caption{}%
    \label{fig:rta_1b}%
  \end{subfigure}
  \caption{$n = 25$, $\sum_{j \in [n-1]} U_j \approx 0.9$}%
\end{figure}

We show various statistics for the number of iterations for $n=25$ and variable
utilization in Table~\ref{tab:rta_1a}. We do the same for variable $n$ and
$\sum_{j \in [n-1]} U_j \approx 0.8$ in Table~\ref{tab:rta_1b}. In all cases,
CP-KERN has better statistics than RTA\@, e.g., smaller means and variances.

\begin{table}[t]
  \centering%
  \begin{tabular}{c c c c c c c}\toprule%
    &\multicolumn{2}{c}{(Min, Max)} &\multicolumn{2}{c}{Mean} &\multicolumn{2}{c}{Variance}\\
    \cmidrule(lr){2-3}%
    \cmidrule(lr){4-5}%
    \cmidrule(lr){6-7}%
    $\sum_{j \in [n-1]} U_j$ &RTA &CP-KERN &RTA &CP-KERN &RTA &CP-KERN\\\midrule
    $0.9$ &$(10, 54)$ &$(4, 27)$ &$23.29$ &$9.29$ &$29.49$ &$7.71$\\
    $0.8$ &$(7, 40)$ &$(3, 23)$ &$14.93$ &$6.91$ &$10.20$ &$3.60$\\
    $0.7$ &$(6, 22)$ &$(3, 15)$ &$11.28$ &$5.68$ &$4.84$ &$2.14$\\\bottomrule
  \end{tabular}
  \caption{Number of iterations for $n=25$ and variable utilization.}%
  \label{tab:rta_1a}%
\end{table}

\begin{table}[t]
  \centering%
  \begin{tabular}{c c c c c c c}\toprule%
    &\multicolumn{2}{c}{(Min, Max)} &\multicolumn{2}{c}{Mean} &\multicolumn{2}{c}{Variance}\\
    \cmidrule(lr){2-3}%
    \cmidrule(lr){4-5}%
    \cmidrule(lr){6-7}%
    $n$ &RTA &CP-KERN &RTA &CP-KERN &RTA &CP-KERN\\\midrule
    $25$ &$(7, 40)$ &$(3, 23)$ &$14.93$ &$6.91$ &$10.20$ &$3.60$\\
    $50$ &$(9, 37)$ &$(4, 23)$ &$17.21$ &$8.82$ &$10.72$ &$5.30$\\
    $75$ &$(11, 42)$ &$(5, 29)$ &$18.60$ &$10.02$ &$11.86$ &$6.85$\\\bottomrule
  \end{tabular}
  \caption{Number of iterations for variable $n$ and $\sum_{j \in [n-1]} U_j
    \approx 0.8$.}%
  \label{tab:rta_1b}%
\end{table}

\subsection{Experiment II}\label{sec:exp_2}

In this experiment, we compare the running time of RTA to the running time of
CP-KERN for synthetic instances of Problem~\eqref{opt:rta}. The initial value is
the same as the one used in the previous experiment.

For $n = 25$ and $\sum_{j \in [n-1]} U_j \approx 0.9$, we show histograms of the
running times of CP-KERN and RTA in Figure~\ref{fig:rta_2a}. The histogram for
CP-KERN is to the left of the histogram for RTA\@; thus, CP-KERN is faster than
RTA\@. Both histograms have long tails, and most instances of the problem are
solved within $20 \mu s$ by both algorithms. We also show a histogram of the
ratio of the CPU time of RTA to the CPU time of CP-KERN in
Figure~\ref{fig:rta_2b}. Unlike, the previous experiment, the ratio is
occasionally below one; thus, RTA is faster than CP-KERN for a small fraction of
the generated instances. This is not an anomaly. In Section~\ref{sec:kern}, we
showed that CP-KERN and RTA have the same asymptotic worst-case running time,
and CP-KERN, unlike RTA, has an optimal convergence rate; thus, theoretically,
CP-KERN is superior to RTA\@. However, the theoretical analysis of the running
time focuses on asymptotic worst-case running times and ignores
implementation-dependent constants; thus, if the convergence rate of the two
algorithms is similar on an instance, there RTA can be faster than CP-KERN\@. On
average, CP-KERN is about $1.4$ times faster than RTA, and the maximum ratio is
about $7$.

\begin{figure}
  \centering%
  \begin{subfigure}[b]{0.75\linewidth}
    \centering%
    \includegraphics[width=\textwidth]{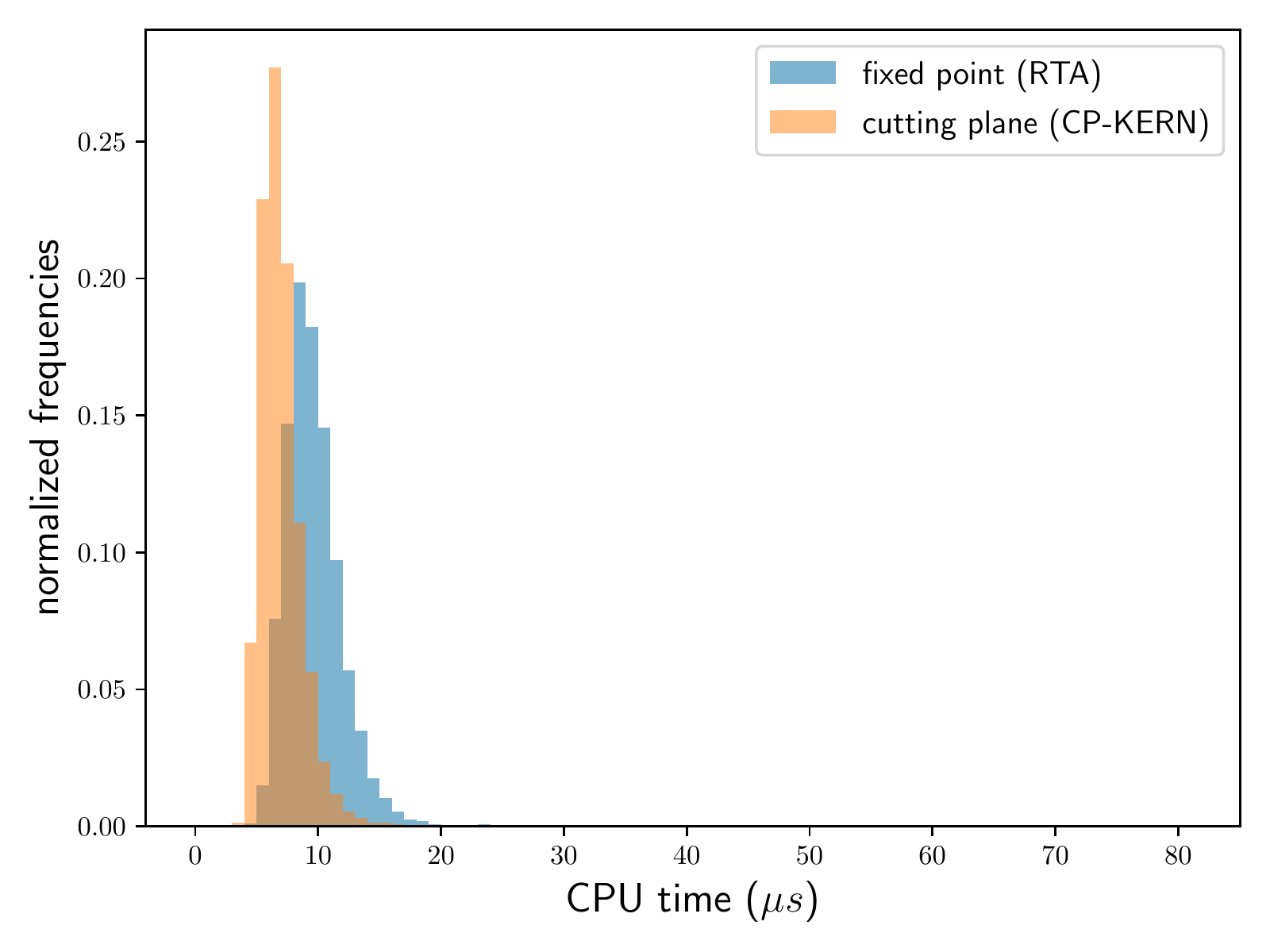}%
    \caption{}%
    \label{fig:rta_2a}%
  \end{subfigure}
  \begin{subfigure}[b]{0.75\linewidth}
    \centering%
    \includegraphics[width=\textwidth]{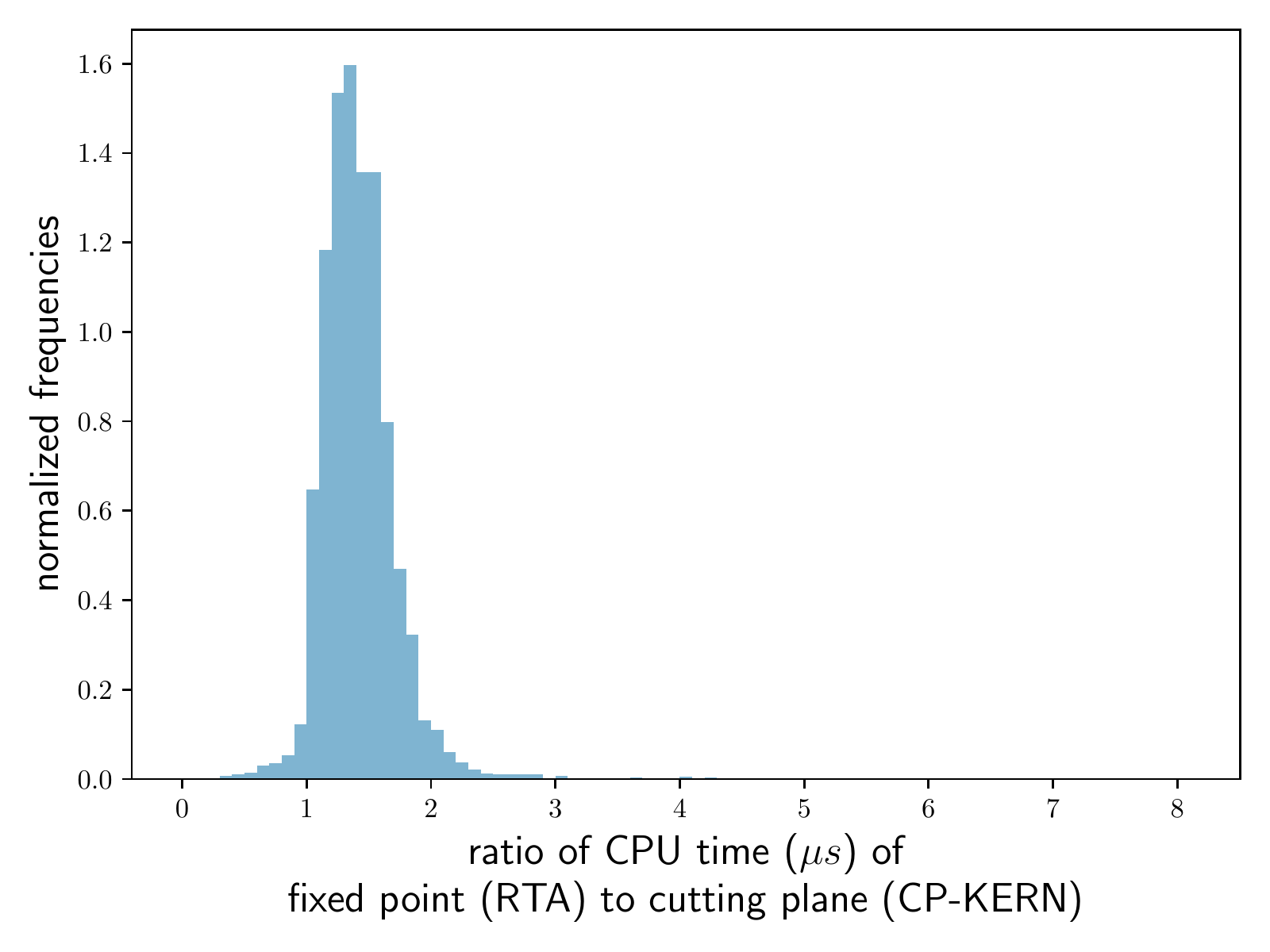}%
    \caption{}%
    \label{fig:rta_2b}%
  \end{subfigure}
  \caption{$n = 25$, $\sum_{j \in [n-1]} U_j \approx 0.9$}%
\end{figure}

We show various statistics for the running times of RTA and CP-KERN for $n=25$
and variable utilization in Table~\ref{tab:rta_2a}. We do the same for variable
$n$ and $\sum_{j \in [n-1]} U_j \approx 0.8$ in Table~\ref{tab:rta_2b}. In all
cases, CP-KERN has better statistics than RTA, e.g, smaller means and variances.

\begin{table}[t]
  \centering%
  \begin{tabular}{c c c c c c c}\toprule%
    &\multicolumn{2}{c}{(Min, Max)} &\multicolumn{2}{c}{Mean} &\multicolumn{2}{c}{Variance}\\
    \cmidrule(lr){2-3}%
    \cmidrule(lr){4-5}%
    \cmidrule(lr){6-7}%
    $\sum_{j \in [n-1]} U_j$ &RTA &CP-KERN &RTA &CP-KERN &RTA &CP-KERN\\\midrule
    $0.9$ &$(4.33, 81.0)$ &$(3.33, 68.0)$ &$9.62$ &$6.94$ &$11.65$ &$6.74$\\
    $0.8$ &$(3.0, 52.67)$ &$(2.67, 51.67)$ &$6.26$ &$5.19$ &$4.84$ &$3.44$\\
    $0.7$ &$(2.67, 41.67)$ &$(2.33, 41.33)$ &$4.85$ &$4.38$ &$2.86$ &$2.13$\\\bottomrule
  \end{tabular}
  \caption{CPU time ($\mu s$) for $n=25$ and variable utilization.}%
  \label{tab:rta_2a}%
\end{table}

\begin{table}[t]
  \centering%
  \begin{tabular}{c c c c c c c}\toprule%
    &\multicolumn{2}{c}{(Min, Max)} &\multicolumn{2}{c}{Mean} &\multicolumn{2}{c}{Variance}\\
    \cmidrule(lr){2-3}%
    \cmidrule(lr){4-5}%
    \cmidrule(lr){6-7}%
    $n$ &RTA &CP-KERN &RTA &CP-KERN &RTA &CP-KERN\\\midrule
    $25$ &$(3.0, 52.67)$ &$(2.67, 51.67)$ &$6.26$ &$5.19$ &$4.84$ &$3.44$\\
    $50$ &$(7.67, 103.67)$ &$(6.67, 109.67)$ &$13.75$ &$12.04$ &$16.29$ &$11.58$\\
    $75$ &$(12.33, 204.33)$ &$(11.67, 163.0)$ &$21.83$ &$20.31$ &$32.93$ &$27.55$\\\bottomrule
  \end{tabular}
  \caption{CPU time ($\mu s$) for variable $n$ and $\sum_{j \in [n-1]} U_j
    \approx 0.8$.}%
  \label{tab:rta_2b}%
\end{table}

\subsection{Experiment III}\label{sec:exp_3}

In this experiment, we compare the number of iterations used by QPA to the
number of iterations used by CP-KERN for synthetic instances of
Problem~\eqref{opt:qpa}. Since the total utilization is less than one in all our
configurations, we use the initial value $L = L_b$ for both algorithms. In
Appendix~\ref{sec:A1}, we show that the bound $L_b$ is implicitly used by
CP-KERN and can also be explicitly passed to CP-KERN by using an alternative
reduction from Problem~\eqref{opt:qpa_2} to the kernel.

For $n = 25$, $\sum_{j \in [n]} U_j \approx 0.9$ and $\sum_{j \in [n]} \delta_j
\approx 1.5$, we show histograms of the number of iterations of CP-KERN and QPA
in Figure~\ref{fig:qpa_1a}. The histogram for CP-KERN is to the left of the
histogram for QPA\@; thus, CP-KERN has a better convergence rate than QPA\@. The
histogram for CP-KERN is thinner and taller than the histogram for QPA\@.
Therefore, the convergence rate of CP-KERN is more predictable than the
convergence rate of QPA\@. We also show a histogram of the ratio of the number
of iterations of QPA to the number of iterations of CP-KERN in
Figure~\ref{fig:qpa_1b}. As expected, the ratio is never below one. The average
ratio is about $2.9$.

\begin{figure}
  \centering%
  \begin{subfigure}[b]{0.75\linewidth}
    \centering%
    \includegraphics[width=\textwidth]{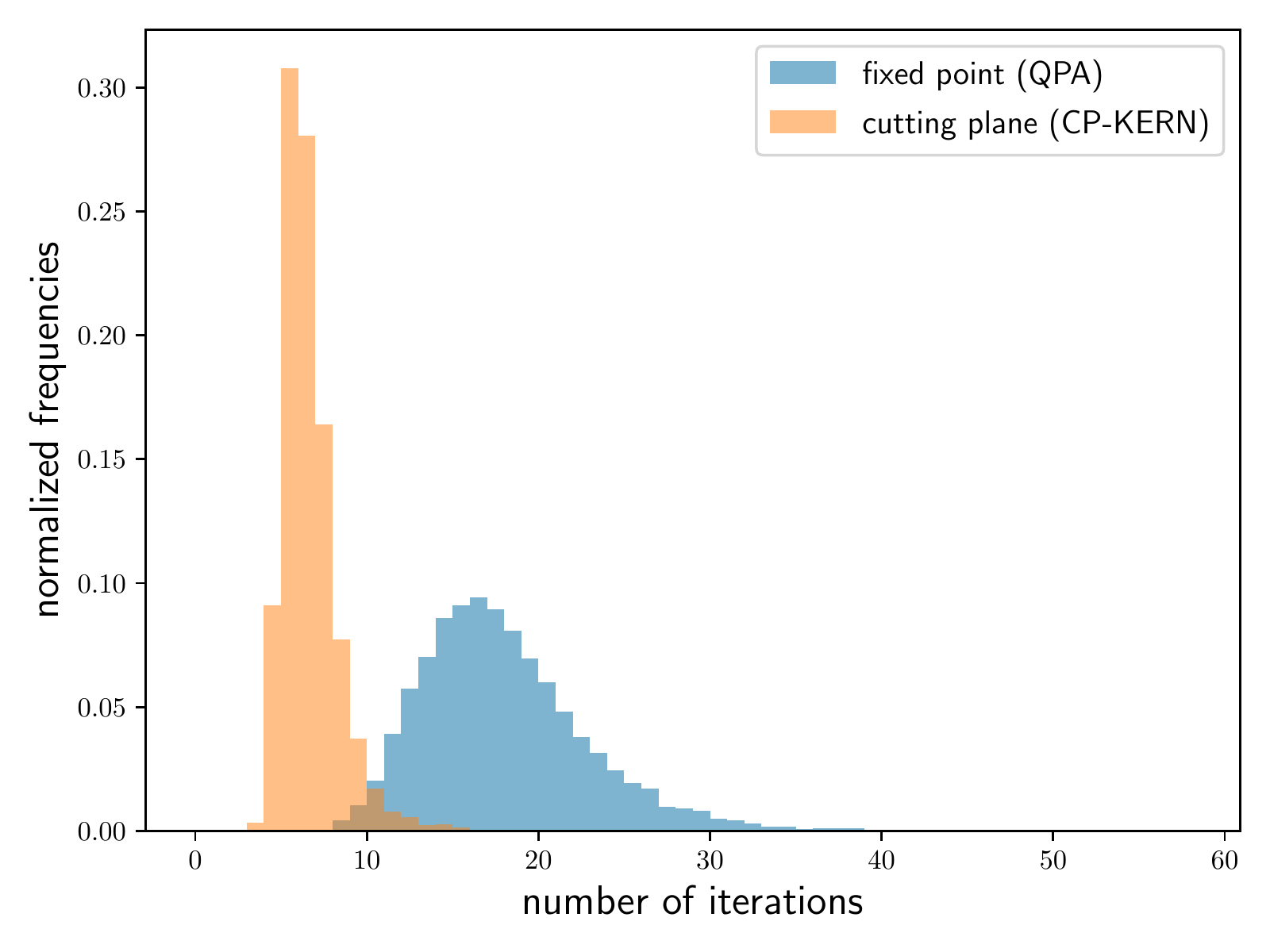}%
    \caption{}%
    \label{fig:qpa_1a}%
  \end{subfigure}
  \begin{subfigure}[b]{0.75\linewidth}
    \centering%
    \includegraphics[width=\textwidth]{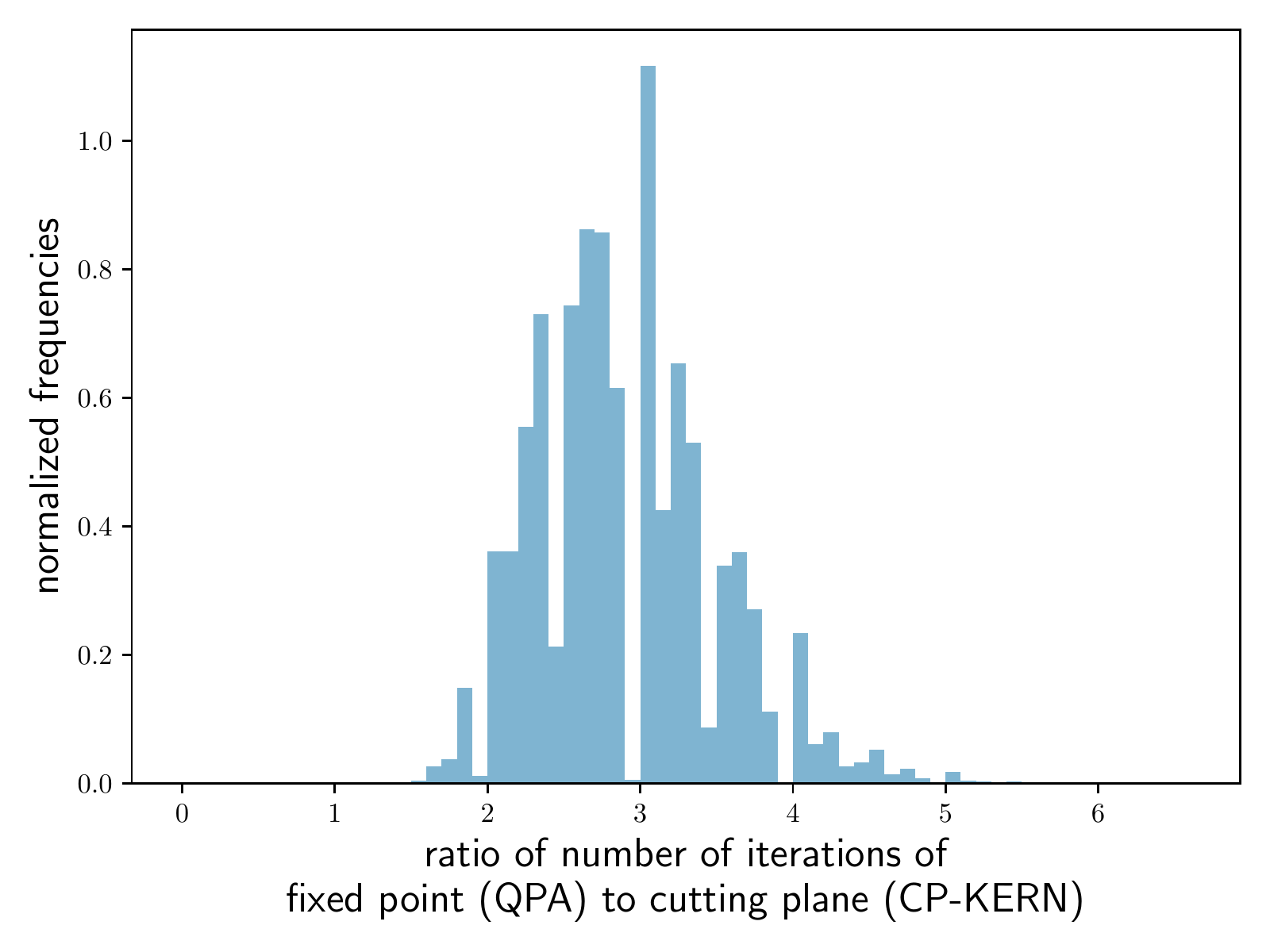}%
    \caption{}%
    \label{fig:qpa_1b}%
  \end{subfigure}
  \caption{$n = 25$, $\sum_{j \in [n]} U_j \approx 0.9$, $\sum_{j \in [n]}
    \delta_j \approx 1.5$}%
\end{figure}

We show various statistics for variable utilization for a fixed configuration in
Table~\ref{tab:qpa_1a}. We do the same for variable density and variable $n$ in
Tables~\ref{tab:qpa_1b} and~\ref{tab:qpa_1c} respectively. In all cases, CP-KERN
has better statistics than QPA, e.g., smaller means and variances.

\begin{table}[t]
  \centering%
  \begin{tabular}{c c c c c c c}\toprule%
    &\multicolumn{2}{c}{(Min, Max)} &\multicolumn{2}{c}{Mean} &\multicolumn{2}{c}{Variance}\\
    \cmidrule(lr){2-3}%
    \cmidrule(lr){4-5}%
    \cmidrule(lr){6-7}%
    $\sum_{j \in [n]} U_j$ &QPA &CP-KERN &QPA &CP-KERN &QPA &CP-KERN\\\midrule
    $0.9$ &$(7, 58)$ &$(3, 20)$ &$17.51$ &$6.14$ &$24.44$ &$2.73$\\
    $0.8$ &$(5, 23)$ &$(3, 11)$ &$10.35$ &$4.51$ &$4.76$ &$0.56$\\
    $0.7$ &$(4, 14)$ &$(2, 6)$ &$7.80$ &$4.02$ &$1.87$ &$0.29$\\\bottomrule
  \end{tabular}
  \caption{Number of iterations for $n=25$, $\sum_{j \in [n]} \delta_j \approx
    1.5$, and variable utilization.}%
  \label{tab:qpa_1a}%
\end{table}

\begin{table}[t]
  \centering%
  \begin{tabular}{c c c c c c c}\toprule%
    &\multicolumn{2}{c}{(Min, Max)} &\multicolumn{2}{c}{Mean} &\multicolumn{2}{c}{Variance}\\
    \cmidrule(lr){2-3}%
    \cmidrule(lr){4-5}%
    \cmidrule(lr){6-7}%
    $\sum_{j \in [n]} \delta_j$ &QPA &CP-KERN &QPA &CP-KERN &QPA &CP-KERN\\\midrule
    $1.25$ &$(5, 30)$ &$(2, 11)$ &$12.74$ &$3.88$ &$10.49$ &$0.47$\\
    $1.50$ &$(7, 58)$ &$(3, 20)$ &$17.51$ &$6.14$ &$24.44$ &$2.73$\\
    $1.75$ &$(9, 67)$ &$(4, 28)$ &$21.61$ &$8.54$ &$38.16$ &$6.33$\\\bottomrule
  \end{tabular}
  \caption{Number of iterations for $n=25$, $\sum_{j \in [n]} U_j \approx
    0.9$, and variable density.}%
  \label{tab:qpa_1b}%
\end{table}

\begin{table}[t]
  \centering%
  \begin{tabular}{c c c c c c c}\toprule%
    &\multicolumn{2}{c}{(Min, Max)} &\multicolumn{2}{c}{Mean} &\multicolumn{2}{c}{Variance}\\
    \cmidrule(lr){2-3}%
    \cmidrule(lr){4-5}%
    \cmidrule(lr){6-7}%
    $n$ &QPA &CP-KERN &QPA &CP-KERN &QPA &CP-KERN\\\midrule
    $25$ &$(7, 58)$ &$(3, 20)$ &$17.51$ &$6.14$ &$24.44$ &$2.73$\\
    $50$ &$(8, 53)$ &$(4, 17)$ &$17.40$ &$6.06$ &$11.52$ &$1.10$\\
    $75$ &$(10, 33)$ &$(4, 14)$ &$17.35$ &$6.05$ &$7.21$ &$0.65$\\\bottomrule
  \end{tabular}
  \caption{Number of iterations for $\sum_{j \in [n]} U_j \approx 0.9$, $\sum_{j
      \in [n]} \delta_j \approx 1.5$, and variable $n$.}%
  \label{tab:qpa_1c}%
\end{table}

\subsection{Experiment IV}\label{sec:exp_4}

In this experiment, we compare the running time of QPA to the running time of
CP-KERN for synthetic instances of Problem~\eqref{opt:qpa}. We use the same
initial value as the previous section.

For $n = 25$, $\sum_{j \in [n]} U_j \approx 0.9$ and $\sum_{j \in [n]} \delta_j
\approx 1.5$, we show histograms of the running times of CP-KERN and QPA in
Figure~\ref{fig:qpa_2a}. The histogram for CP-KERN is to the left of the
histogram for QPA\@; thus, CP-KERN is faster than QPA\@. Similar to Experiment
II, most instances of the problem are solved within $20 \mu s$ by both
algorithms. We also show a histogram of the ratio of the running time of QPA to
the running time of CP-KERN in Figure~\ref{fig:qpa_2b}. QPA is faster than
CP-KERN for a small fraction of the generated instances, but the average ratio
is about $1.3$, and the maximum ratio is about $8$.

\begin{figure}
  \centering%
  \begin{subfigure}[b]{0.75\linewidth}
    \centering%
    \includegraphics[width=\textwidth]{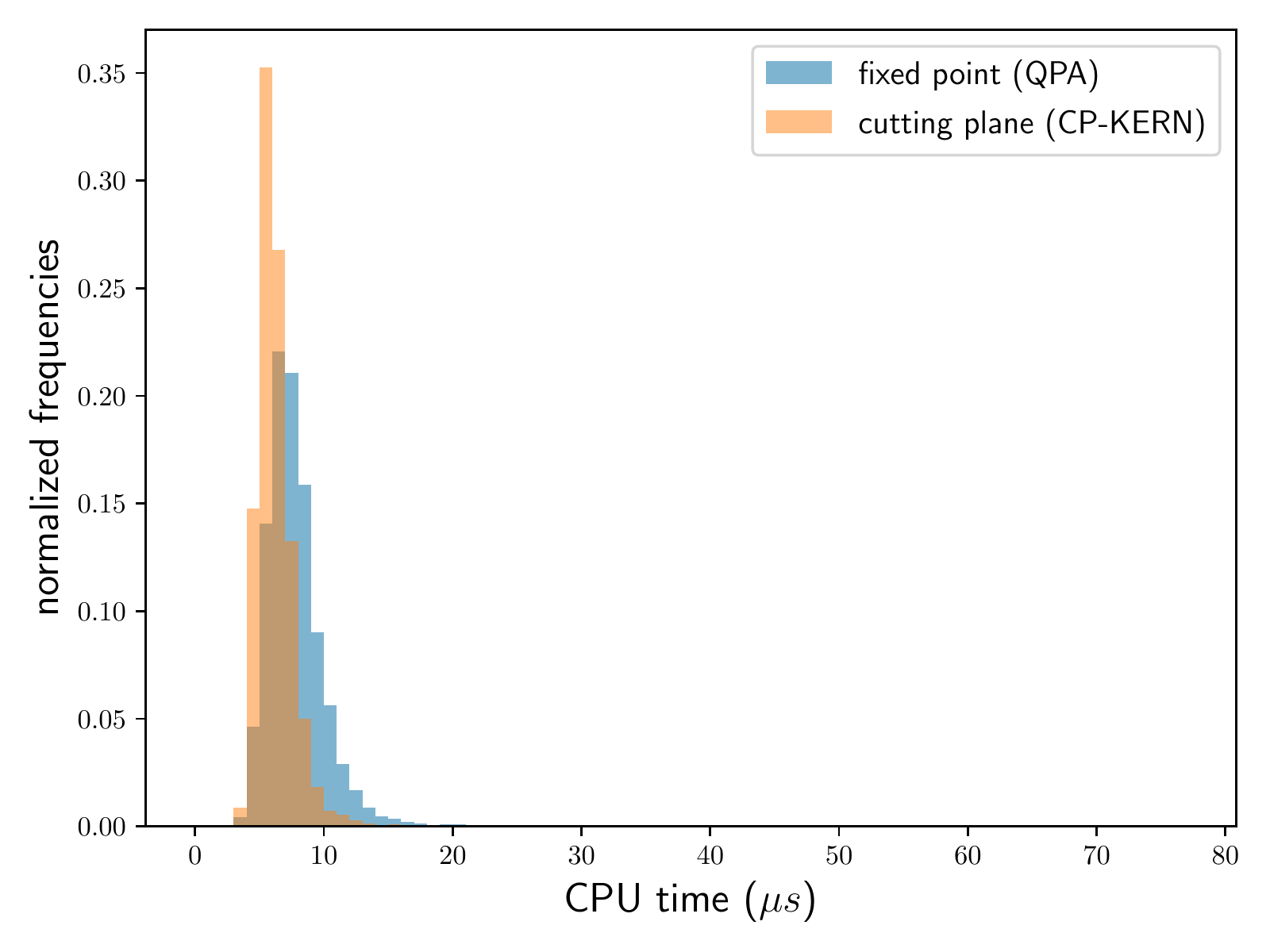}%
    \caption{}%
    \label{fig:qpa_2a}%
  \end{subfigure}
  \begin{subfigure}[b]{0.75\linewidth}
    \centering%
    \includegraphics[width=\textwidth]{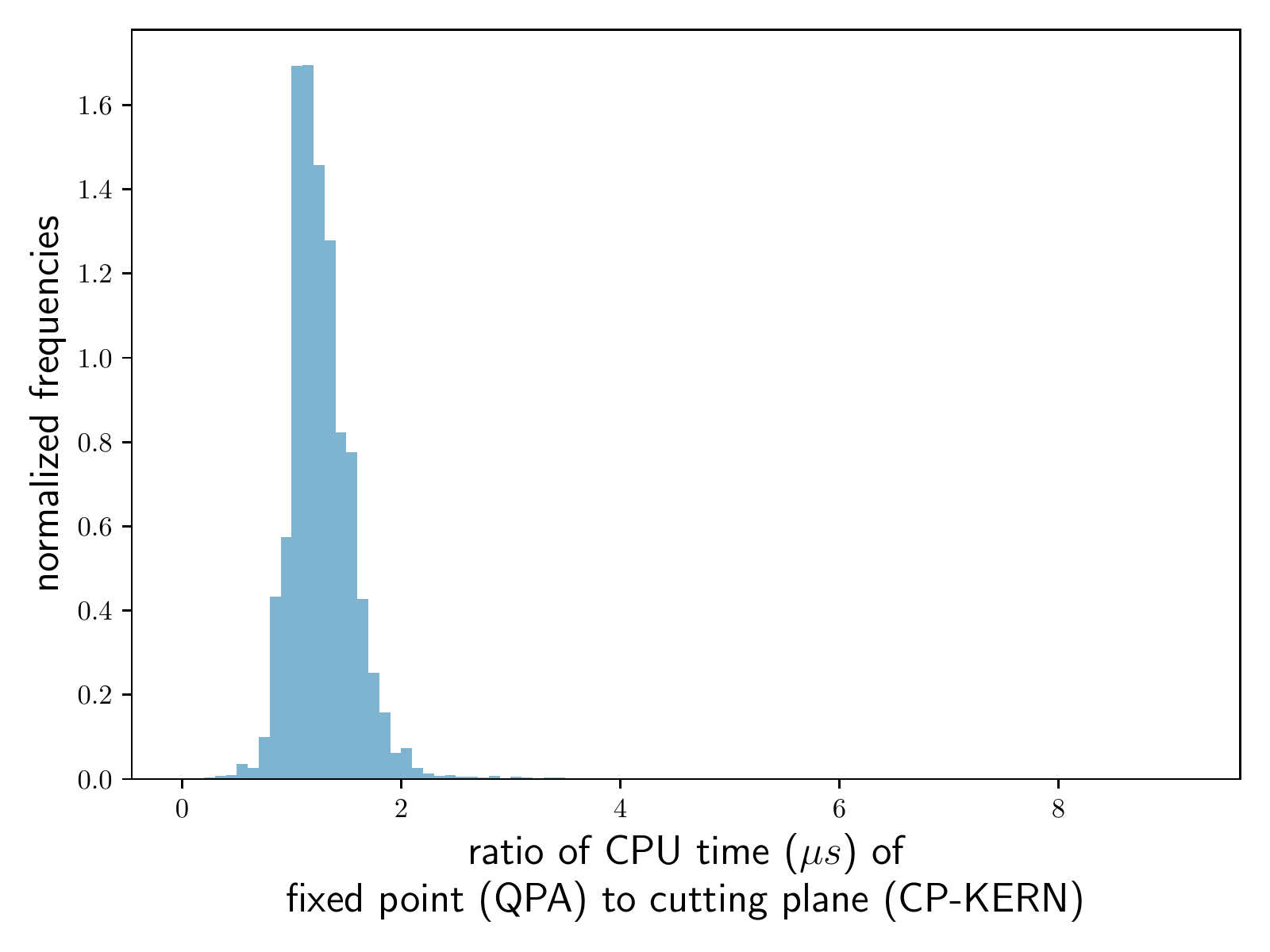}%
    \caption{}%
    \label{fig:qpa_2b}%
  \end{subfigure}
  \caption{$n = 25$, $\sum_{j \in [n]} U_j \approx 0.9$, $\sum_{j \in [n]}
    \delta_j \approx 1.5$}%
\end{figure}

We show various statistics for variable utilization for a fixed configuration in
Table~\ref{tab:qpa_2a}. We do the same for variable density and variable $n$ in
Tables~\ref{tab:qpa_2b} and~\ref{tab:qpa_2c} respectively. From the last rows of
Tables~\ref{tab:qpa_2a} and~\ref{tab:qpa_2c}, we can see that QPA is faster than
CP-KERN for some configurations where the system utilization is small or $n$ is
large. When the system utilization is small, then the denominator of $L_b$ is
large, and hence $L_b$ is small; thus, the instances can be expected to be less
challenging for both the algorithms (recall that $L$ occurs in the numerator in
the asymptotic characterization of worst-case running times in
Theorems~\ref{thm:qpa} and~\ref{thm:edf}). In Experiment II, we were able to
control the hardness of the instances to some extent by choosing $D_n$ to be a
large constant. Since we do not exercise any such control here, less challenging
instances are produced in some cases. This explains why the simpler QPA
outperforms the theoretically superior CP-KERN in the last configuration in
Table~\ref{tab:qpa_2a}. A similar explanation works for the last configuration
in Table~\ref{tab:qpa_2c} as well. When $n = 75$ and the total density is
$1.75$, the deadlines are quite close to the periods, and hence the numerator of
$L_b$ is small; thus, less challenging instances are generated. We have verified
that CP-KERN is faster than QPA for $n = 75$ and larger densities such as $5$.

\begin{table}[t]
  \centering%
  \begin{tabular}{c c c c c c c}\toprule%
    &\multicolumn{2}{c}{(Min, Max)} &\multicolumn{2}{c}{Mean} &\multicolumn{2}{c}{Variance}\\
    \cmidrule(lr){2-3}%
    \cmidrule(lr){4-5}%
    \cmidrule(lr){6-7}%
    $\sum_{j \in [n]} U_j$ &QPA &CP-KERN &QPA &CP-KERN &QPA &CP-KERN\\\midrule
    $0.9$ &$(3.0, 76.67)$ &$(3.0, 50.0)$ &$7.67$ &$6.12$ &$8.51$ &$4.35$\\
    $0.8$ &$(2.33, 39.33)$ &$(2.67, 40.33)$ &$4.65$ &$4.58$ &$2.26$ &$2.03$\\
    $0.7$ &$(1.67, 29.67)$ &$(2.33, 39.0)$ &$3.64$ &$4.04$ &$1.25$ &$1.56$\\\bottomrule
  \end{tabular}
  \caption{Running time for $n=25$, $\sum_{j \in [n]} \delta_j \approx 1.5$, and
    variable utilization.}%
  \label{tab:qpa_2a}%
\end{table}

\begin{table}[t]
  \centering%
  \begin{tabular}{c c c c c c c}\toprule%
    &\multicolumn{2}{c}{(Min, Max)} &\multicolumn{2}{c}{Mean} &\multicolumn{2}{c}{Variance}\\
    \cmidrule(lr){2-3}%
    \cmidrule(lr){4-5}%
    \cmidrule(lr){6-7}%
    $\sum_{j \in [n]} \delta_j$ &QPA &CP-KERN &QPA &CP-KERN &QPA &CP-KERN\\\midrule
    $1.25$ &$(2.33, 56.67)$ &$(2.33, 42.0)$ &$5.64$ &$4.39$ &$3.95$ &$2.01$\\
    $1.50$ &$(3.0, 76.67)$ &$(3.0, 50.0)$ &$7.67$ &$6.12$ &$8.51$ &$4.35$\\
    $1.75$ &$(4.0, 121.67)$ &$(3.0, 87.67)$ &$9.18$ &$7.42$ &$13.77$ &$8.03$\\\bottomrule
  \end{tabular}
  \caption{Running time for $n=25$, $\sum_{j \in [n]} U_j \approx
    0.9$, and variable density.}%
  \label{tab:qpa_2b}%
\end{table}

\begin{table}[t]
  \centering%
  \begin{tabular}{c c c c c c c}\toprule%
    &\multicolumn{2}{c}{(Min, Max)} &\multicolumn{2}{c}{Mean} &\multicolumn{2}{c}{Variance}\\
    \cmidrule(lr){2-3}%
    \cmidrule(lr){4-5}%
    \cmidrule(lr){6-7}%
    $n$ &QPA &CP-KERN &QPA &CP-KERN &QPA &CP-KERN\\\midrule
    $25$ &$(3.0, 76.67)$ &$(3.0, 50.0)$ &$7.67$ &$6.12$ &$8.51$ &$4.35$\\
    $50$ &$(6.33, 132.33)$ &$(8.0, 122.67)$ &$14.26$ &$13.91$ &$16.67$ &$12.97$\\
    $75$ &$(11.67, 163.67)$ &$(14.67, 178.0)$ &$20.66$ &$22.22$ &$19.41$ &$18.29$\\\bottomrule
  \end{tabular}
  \caption{Running time for $\sum_{j \in [n]} U_j \approx 0.9$, $\sum_{j \in
      [n]} \delta_j \approx 1.5$, and variable $n$.}%
  \label{tab:qpa_2c}%
\end{table}

\section{Conclusion}%
\label{sec:conc}%

We have achieved a logical unification of four different algorithms (RTA, IP-FP,
QPA, IP-EDF) by showing that they all belong to a family of cutting-plane
algorithms. CP-KERN has the optimal convergence rate in the family; RTA and QPA
have suboptimal convergence rates. In theory, CP-KERN, RTA and QPA have the same
worst-case running times. The empirical evaluations show that
\begin{itemize}
\item CP-KERN has higher convergence rates than RTA and QPA for randomly
  generated systems.
\item CP-KERN has smaller running times than RTA and QPA for randomly generated
  systems that are hard.
\end{itemize}
Unlike the convergence rates, the running times can vary a lot depending on
low-level implementation choices.

For any given schedulability problem instance, the number of iterations required
for convergence is a good machine-independent measure of the hardness of the
instance, especially for methods that proceed by estimating dual bounds. Using
the number of iterations as a measure of the hardness of a schedulability
instance is similar to counting the number of recursive calls to the
Davis-Putnam (DP) procedure for SAT
instances~\citep{selmanGeneratingHardSatisfiability1996}, especially when one
considers the connections between DP, resolution, and cutting
planes~\citep{chvatalEdmondsPolytopesHierarchy1973,
  cookComplexityCuttingplaneProofs1987, hookerGeneralizedResolutionCutting1988}.
We plan to investigate these connections in future work to understand how to
generate hard schedulability instances.

Our cutting planes for schedulability problems are quite similar to textbook
cutting planes like Gomory cuts. A different approach towards generating cuts is
the so-called \emph{polyhedral approach} wherein facet-defining inequalities are
used instead of traditional textbook cutting
planes~\citep{padbergBranchandCutAlgorithmResolution1991}. The polyhedral
approach has been used successfully to understand scheduling models in the OR
community~\citep{queyrannePolyhedralApproachesMachine1996,
  vandenakkerPolyhedralApproachSinglemachine1999}. In future work, we hope to
apply the polyhedral approach to schedulability problems.

Most efforts to improve the running times of the fixed-point iteration tests
have focused on finding a good initial guess for the fixed
point~\citep{sjodinImprovedResponsetimeAnalysis1998,
  brilInitialValuesOnline2003, davisEfficientExactSchedulability2008}. Since we
have shown that fixed-point iteration tests are special cases of cutting-plane
algorithms, ideas like preprocessing and branch-and-cut can be experimented with
to speed up schedulability tests. We expect branch-and-cut adaptations of
CP-KERN to be much faster than CP-KERN, RTA, and QPA, especially if the tests
can utilize multiple processors and if we are satisfied with feasible solutions
to Problems~\eqref{opt:rta} and~\eqref{opt:qpa}.

\subsection*{Acknowledgments}

The author would like to thank anonymous reviewers for their valuable
suggestions. The author would also like to thank Pontus Ekberg and Sanjoy Baruah
for discussions on the subject.

\subsection*{Declarations}

This work is supported by the US National Science Foundation under Grant numbers
CPS-1932530 and CNS-2141256.

\bibliographystyle{plainnat}%
\bibliography{ref}%

\appendix%

\section{Properties of $f$}%
\label{sec:A0}%

We establish basic properties of $f$ (Definition~\eqref{def:f}) in the context
of the kernel. We also assume that the vectors $T$, $U$, $\alpha$, and $\u{x}$
are sorted in a nonincreasing order using the key $T_j\u{x}_j - \alpha_j$ for
each $j \in [n]$.

\begin{lem}\label{lem:f_id}
  For any $k \in [n]$, if $f(k-1)$ is well-defined, then we must have
  \[
    \frac{f(k) - f(k-1)}{U_k} = \frac{\u{x}_kT_k - \alpha_k - f(k)}{1 - \sum_{j
        \in [n]\setminus[k-1]}U_j} = \frac{\u{x}_kT_k - \alpha_k - f(k-1)}{1 -
      \sum_{j \in [n]\setminus[k]}U_j}.
  \]
  From the identities, it follows that
  \[
    \sgn(f(k) - f(k-1)) = \sgn(\u{x}_kT_k - \alpha_k - f(k)) = \sgn(\u{x}_kT_k -
    \alpha_k - f(k-1)).
  \]
\end{lem}
\begin{proof}[Proof Sketch]
  Consider the following equations:
  \begin{align*}
    &(1 - \textstyle\sum_{j \in [n]\setminus[k-1]}U_j)(f(k) - f(k-1))\\
    = &(1 - \textstyle\sum_{j \in [n]\setminus[k]}U_j)f(k) - (1 - \textstyle\sum_{j \in [n]\setminus[k-1]}U_j)f(k-1) - U_kf(k)\\
    = &(\beta + \textstyle\sum_{j \in [n] \setminus [k]} \alpha_jU_j + \textstyle\sum_{j \in [k]}
        \u{x}_jC_j) - \\
    &(\beta + \textstyle\sum_{j \in [n] \setminus [k-1]} \alpha_jU_j + \textstyle\sum_{j \in [k-1]}
      \u{x}_jC_j) - U_kf(k)\tag{using Definition~\eqref{def:f}}\\
    = &\u{x}_kC_k - \alpha_kU_k - U_kf(k)\\
    = &U_k(\u{x}_kT_k - \alpha_k - f(k))
  \end{align*}
  From the definition of the kernel and the well-definedness of $f(k-1)$, we
  have $U_k > 0$ and $1 - \sum_{j \in [n]\setminus[k-1]}U_j > 0$. Combining
  these facts with the above equations gives
  \[
    \sgn(f(k) - f(k-1)) = \sgn(\u{x}_kT_k - \alpha_k - f(k)).
  \]
  Similarly, we can show that
  \[
    (1 - \textstyle\sum_{j \in [n]\setminus[k]}U_j)(f(k) - f(k-1)) =
    U_k(\u{x}_kT_k - \alpha_k - f(k-1)),
  \]
  and the statement about the signs of various expressions follows from the
  positivity of the excluded expressions.
\end{proof}

\begin{lem}\label{lem:f_local_min}
  If $f$ has a local minimum point at $i$ in the interior of its domain, then we
  must have
  \[
    f(i-1) = f(i) = f(i+1).
  \]
\end{lem}
\begin{proof}[Proof sketch]
  Assume that $i$ is a local minimum point of $f$ in the interior of its domain.
  From Lemma~\ref{lem:f_id}, we get
  \[
    \u{x}_iT_i - \alpha_i \le f(i) \le \u{x}_{i+1}T_{i+1} - \alpha_{i+1}.
  \]
  From the order imposed on the vectors by sorting, it follows that
  $\u{x}_{i+1}T_{i+1} - \alpha_{i+1} \le \u{x}_iT_i - \alpha_i$. Combining the
  three inequalities gives
  \[
    \u{x}_iT_i - \alpha_i = f(i) = \u{x}_{i+1}T_{i+1} - \alpha_{i+1}.
  \]
  Using Lemma~\ref{lem:f_id} again, we must have
  \[
    f(i-1) = f(i) = f(i+1).
  \]
\end{proof}

The absence of strict local minimum points in the interior of $f$ trivially
implies the following corollary.

\begin{cor}\label{cor:f_local_global}
  If $f$ has two local maximum points $i_1$ and $i_2$ in $\{0\} \cup [n]$ with
  $i_1 \le i_2$, then we must have
  \[
    \forall i \in [i_2] \setminus [i_1]: f(i) = f(i_1).
  \]
  In other words, any local maximum point of $f$ is a global maximum point of $f$,
  and $f$ is constant between two local maximum points.
\end{cor}

\section{On initial values of fixed-point iteration approaches}%
\label{sec:A1}%

Problem~\eqref{opt:rta} can be reduced to the kernel by initializing
$(n,\alpha,\beta,a,b)$ in the target instance to $(n,J,0,1,D_n-J_n)$. Recall
that $\rbf$ is given by
\[
  t \mapsto \sum_{j \in [n]} \left\lceil \frac{t + J_j}{T_j} \right\rceil C_j.
\]
Since we are only concerned about the behavior of $\rbf$ for the domain
$\interval{1}{D_n-J_n}$, we must have
\[
  t + J_n \le D_n \le T_n
\]
The second inequality is true because we have assumed that the FP system has
constrained deadlines. Thus, $\ceil{(t+J_n)/T_n} = 1$ and $\rbf$ can be
simplified to
\[
  t \mapsto C_n + \sum_{j \in [n-1]} \left\lceil \frac{t + J_j}{T_j} \right\rceil C_j.
\]
The schedulability condition can be written as
\begin{equation}%
  \label{cond:fp_2}%
  \exists t \in \interval{C_n}{D_n-J_n}: C_n + \sum_{j \in [n-1]} \left\lceil \frac{t + J_j}{T_j} \right\rceil C_j \le t.
\end{equation}
Using this schedulability condition, we can reduce Problem~\eqref{opt:rta}
to the kernel by initializing $(n,\alpha,\beta,a,b)$ in the target instance to
$(n-1,J,C_n,C_n,D_n-J_n)$.

When we solve the target instance using CP-KERN, $f$ is given by
\begin{equation}%
  \label{def:f_fp}%
  k \mapsto \frac{C_n + \sum_{j \in [n-1] \setminus [k]} J_jU_j + \sum_{j \in [k]} \u{x}_jC_j}{1 - \sum_{j \in [n-1] \setminus [k]}U_j}
\end{equation}
In each iteration of CP-KERN, the value $t^* = \max f$ is at least
\begin{equation}
  f(0) = \frac{C_n + \sum_{j \in [n-1]} J_jU_j}{1 - \sum_{j \in [n-1]}U_j}
\end{equation}
$f(0)$ has been proposed as an initial value for RTA in prior work by using a
different analysis method~\citep{sjodinImprovedResponsetimeAnalysis1998}.
However, we do not need to pass this initial bound to CP-KERN since it finds it
implicitly; on the other hand, initializing $(n,\alpha,\beta,a,b)$ in the target
instance to $(n-1, J, C_n, \ceil{f(0)}, D_n-J_n)$ is a valid reduction from
Problem~\eqref{opt:rta} to the kernel.

Problem~\eqref{opt:qpa_2} (the $k$-th subproblem of EDF schedulability)
can be solved by reducing it to the kernel by setting $(n,\alpha,\beta,a,b)$ to
$(k,D-T-J, 1, -b_k+1, -a_k)$. When we solve the target instance by using
CP-KERN, $f$ is given by
\begin{equation}%
  \label{def:f_k}%
  \ell \mapsto \frac{1 + \sum_{j \in [k] \setminus [\ell]} (\hat{D}_j-T_j)U_j +
    \sum_{j \in [\ell]} \u{x}_jC_j}{1 - \sum_{j \in [k] \setminus
      [\ell]}U_j}
\end{equation}
In each iteration of CP-KERN, the value $t^* = \max f$ is at least
\begin{equation}
  f(0) = \frac{\sum_{j \in [k]} (\hat{D}_j-T_j)U_j - 1}{1 - \sum_{j \in [k]}U_j}
\end{equation}
Since the sign of $t$ was inverted in the reduction to the kernel, $-f(0)$ is an
upper bound for $t$ in Problem~\eqref{opt:qpa_2}. It is not too hard to see that
the upper bound $L_b$ (see Definition~\eqref{def:Lb}) is implicitly present in
$-f(0)$ when $k = n$. Since $f(0)$ is a lower bound for $t$ in the kernel, when
reducing Problem~\eqref{opt:qpa_2} to the kernel, $(n,\alpha,\beta,a,b)$ in the
target instance can be initialized to $(k,\hat{D}-T, 1, \max(-b_k+1,
\ceil{f(0)}), -a_k)$.

\section{IP formulations for asynchronous EDF schedulability}%
\label{sec:A3}%

If the first request for a periodic task $i$ arrives at time $O_i$, then $O_i$
is called the \emph{phase} or \emph{arrival offset} of task $i$. If all tasks in
a system have equal phases (in which case they are typically assumed to be equal
to zero and omitted from the specifications), then it is called a
\emph{synchronous system}; otherwise, it is called an \emph{asynchronous
  system}. In this section, we assume that the tasks are periodic and
asynchronous with zero release jitter. Such a system is unschedulable if and
only if $\sum_{i \in [n]} U_i > 1$ or
\begin{equation}%
  \label{cond:eta}%
    \exists t_1, t_2 \in \interval{0}{O_{\max} + 2\hyp} \cap \Z:\ \sum_{i \in [n]} \eta_i(t_1,t_2) C_i > t_2 - t_1 \ge 0
\end{equation}
holds~\citep[][Lem. 3.4]{baruahAlgorithmsComplexityConcerning1990}, where
\begin{equation}%
  \label{def:eta}%
  \eta_i(t_1, t_2) = \lvert \{k \in \N \mid t_1 \le O_i + kT_i, O_i + kT_i + D_i \le t_2\} \rvert
\end{equation}
Thus, $\eta_i(t_1, t_2)$ is the number of times task $i$ has both its release
time and deadline in $\interval{t_1}{t_2}$.

For such asynchronous systems, \citet{baruahAlgorithmsComplexityConcerning1990}
formulate (informally) Condition~\eqref{cond:eta} as an integer program in Thm.
3.5. However, the formulation is flawed: this can be verified by applying the
formulation on the system in Table~\ref{tab:ex2}.

$\eta_i$ can be simplified by considering three cases. If $t_1 \le O_i$ and $O_i
+ D_i \le t_2$, then we have
\[
  \eta_i(t_1, t_2) = \left\lfloor\frac{t_2 + T_i - D_i - O_i}{T_i}\right\rfloor.
\]
Otherwise, if $t_2 > O_i$ and $t_2 - t_1 \ge D_i$, then we have
\[
  \eta_i(t_1, t_2) = \left\lfloor\frac{t_2 + T_i - D_i - O_i}{T_i}\right\rfloor
  - \left\lceil\frac{t_1 - O_i}{T_i}\right\rceil.
\]
Otherwise, we have
\[
  \eta_i(t_1, t_2) = 0.
\]
A branching strategy, like the one we used in Section~\ref{sec:ip_edf}, follows
straightforwardly from the three conditions, resulting in $\bigO(n^3)$
subproblems with simple IP formulations.

\end{document}